\documentclass[smallextended]{svjour3}

\usepackage[T1]{fontenc}

\usepackage{caption}
\usepackage{amsfonts}
\usepackage{amssymb}
\usepackage{amsmath}
\usepackage{mathrsfs}
\usepackage{dsfont}

\usepackage{algorithm}
\usepackage{algorithmicx}
\usepackage{algpseudocode}

\usepackage{booktabs}
\usepackage{afterpage}

\usepackage{graphicx}
\usepackage{subcaption}

\usepackage{xcolor}
\usepackage{url}

\DeclareMathOperator*{\argmax}{arg\,max}
\DeclareMathOperator*{\diag}{diag}
\DeclareMathOperator{\vect}{vec}
\renewcommand{\vec}[1]{\mathbf{#1}}
\newcommand{\mat}[1]{\mathbf{#1}}
\newcommand{\abs}[1]{\lvert #1 \rvert}
\newcommand{\Var}{\mathrm{Var}}
\newcommand{\qedsymbol}{\rightline{$\square$}}

\begin{document}

\title{A fast Monte Carlo algorithm for evaluating matrix functions with application in complex networks}
\titlerunning{Fast Monte Carlo algorithm for evaluating matrix functions}
\author{Nicolas L. Guidotti \and Juan A. Acebr\'on \and Jos\'e Monteiro}

\institute{Nicolas L. Guidotti \at INESC-ID, Instituto Superior T\'ecnico, Universidade de Lisboa, Portugal, \\\email{nicolas.guidotti@tecnico.ulisboa.pt}  \and
Juan A. Acebr\'on \at Department of Mathematics, Carlos III University of Madrid, Spain, and \\
INESC-ID, Instituto Superior T\'ecnico, Universidade de Lisboa, Portugal,
\\\email{juan.acebron@ist.utl.pt} \and
Jos\'e Monteiro \at INESC-ID, Instituto Superior T\'ecnico, Universidade de Lisboa, Portugal, \\\email{ jcm@inesc-id.pt}}

\date{Received: date / Accepted: date}

\maketitle


\begin{abstract}
 We propose a novel stochastic algorithm that randomly samples entire rows and columns of the matrix as a way to approximate an arbitrary matrix function using the power series expansion. This contrasts with existing Monte Carlo methods, which only work with one entry at a time, resulting in a significantly better convergence rate than the original approach. To assess the applicability of our method, we compute the subgraph centrality and total communicability of several large networks. In all benchmarks analyzed so far, the performance of our method was significantly superior to the competition, being able to scale up to 64 CPU cores with remarkable efficiency.


\begin{acknowledgements}
This work was supported by national funds through FCT, Fundação para a Ciência e a Tecnologia, under projects URA-HPC PTDC/08838/2022, UIDB/50021/2020 (DOI:10.54499/UIDB/50021/2020) and PTDC/CCI-INF/6762/2020 and grant 2022.11506.BD.
JA was funded by  Ministerio de Universidades and specifically the requalification program of the Spanish University System 2021-2023 at the Carlos III University.
\end{acknowledgements}

\end{abstract}

\clearpage


\section{Introduction}\label{sec:intro}
Complex networks have become an essential tool in many scientific areas~\cite{estrada_structure_2012,newman_networks_2018,newman_structure_2003} for studying the interaction between different elements within a system. From these networks, it is possible, for instance, to identify the most important elements in the system, such as the key proteins in a protein interaction network~\cite{estrada_virtual_2006}, the keystone species in an ecological network~\cite{jordan_quantifying_2007}, vulnerable infrastructures~\cite{albert_error_2000} or the least resilient nodes in a transportation network~\cite{aparicio_lines_2022}.

Measures of node importance are referred to as \textit{node centrality} and many metrics have been proposed over the years \cite{estrada_structure_2012,newman_networks_2018}. Here, we consider the family of centrality measures defined in terms of matrix functions \cite{estrada_network_2010}, which classify the nodes according to how well they can spread information to other nodes in the network. Both the Katz centrality~\cite{katz_new_1953} and subgraph centrality~\cite{estrada_subgraph_2005} belong to this family. In most cases, the node centrality is computed based on the matrix resolvent $(\mat{I} - \gamma\mat{A})^{-1}$ and the exponential $\exp{(\mat{A})}$, but other functions~\cite{arrigo_ml_network_2021,benzi_ranking_2013} can be used as well,
with $\mat{A} \in \mathds{R}^{n \times n}$ denoting the network's adjacency matrix. In this paper, we focus on the subgraph centrality \cite{estrada_subgraph_2005} (i.e., $\diag(f(\mat{A}))$)  and total communicability~\cite{benzi_total_2013} (i.e., $f(\mat{A})\vec{1}$) problems, where $f(\mat{A}))$ is a given matrix function.

Although these centrality measures have been successfully used in many problems~\cite{benzi_ranking_2013,estrada_physics_2012,estrada_network_2010,estrada_subgraph_2005}, computing a matrix function for large networks can be very demanding using the current numerical methods. Direct methods, such as \texttt{expm}~\cite{al-mohy_new_2010,higham_scaling_2005} or the Schur-Parlett algorithm~\cite{davies_schur-parlett_2003}, have a computational cost of $O(n^3)$ and yield a full dense matrix $f(\mat{A})$, hence are only feasible for small matrices. Methods based on Gaussian quadrature rules~\cite{benzi_quadrature_2010,fenu_block_2013,golub_matrices_2009} can estimate the diagonal entries of $f(\mat{A})$ without evaluating the full matrix function, but are prone to numerical breakdown when $\mat{A}$ is sparse and large (which is often the case with real networks), and thus, are often employed to determine only the most important nodes in the network. Krylov-based methods \cite{afanasjew_implementation_2008,eiermann_restarted_2006,guttel_rational_2013,guttel_limitedmemory_2020} can efficiently compute $f(\mat{A})\vec{v}$, for $\vec{v} \in \mathbb{R}^{n}$, provided that $\mat{A}$ is well conditioned. Otherwise, their convergence rate can be very slow or even stagnate since a general and well-established procedure to precondition the matrix $\mat{A}$ does not exist. Rational Krylov methods are often more resilient to the condition number and provide a better approximation to $f(\mat{A})\vec{v}$ than polynomial ones, but require solving a linear system for each vector of the basis. Moreover, the stopping criteria for these methods remains an open issue~\cite{guttel_limitedmemory_2020}.

Monte Carlo methods \cite{benzi_analysis_2017,dimov_monte_2008,dimov_parallel_2001,forsythe_matrix_1950,ji_convergence_2013} provide an alternative way to calculate matrix functions, primarily for solving linear systems. In essence, these methods construct a discrete-time Markov chain whose underlying random paths evolve through the different indices of the matrix, which can be formally understood as the Monte Carlo sampling of the corresponding Neumann series. Their convergence has been rigorously established in \cite{benzi_analysis_2017,dimov_new_2015,ji_convergence_2013}. Recently, \cite{acebron_monte_2019,acebron_highly_2020,guidotti_stochastic_2023} extended these methods for the evaluation of the matrix exponential and Mittag-Leffler functions.

Another strategy is to construct a random \textit{sketch} (i.e., a probabilistic representation of the matrix) and then use it to approximate the desired operation. This is a basic idea in contemporary numerical linear algebra \cite{martinsson_randomized_2020,murray_randomized_2023}. Some recent studies have shown that a polynomial Krylov method can be accelerated using randomization techniques \cite{cortinovis_speeding_2023,guttel_randomized_2023}.
In this paper, we propose a new stochastic algorithm that randomly samples full rows and columns of the matrix as a way to approximate the target function using the corresponding power series expansion. Through an extensive set of numerical experiments, we show that our approach converges much faster than the original Monte Carlo method and that it is particularly effective for estimating the subgraph centrality and total communicability of large networks. We also compare our method against other classical and randomized methods considering very large matrices.

The paper is organized as follows. Section~\ref{sec:background} presents a brief overview of the centrality measures defined in terms of matrix functions. Section~\ref{sec:algorithm} describes our randomized algorithm and how it can be implemented efficiently. In Section~\ref{sec:results}, we evaluate the performance and accuracy of our method by running several benchmarks with both synthetic and real networks. We also compare our method against several other algorithms. In the last section, we conclude our paper and present some future work.

\section{Background} \label{sec:background}

In this section, we introduce some definitions and ideas from graph theory that will be used throughout the paper.

\subsection{Graph definitions} \label{sec:graph_def}

A \textit{graph} or \textit{network} $G = (V, E)$ is composed of a set $V = \{1, 2, \dots, n\}$ of \textit{nodes} (or \textit{vertices}) and a set $E = \{(u, v) : u, v \in V\}$ of \textit{edges} between them~\cite{newman_networks_2018}. A graph is \textit{undirected} if the edges are bidirectional and \textit{directed} if the edges are unidirectional. A \textit{walk} of length $k$ over the graph $G$ is a sequence of vertices $v_1, v_2, \dots, v_{k + 1}$ such that $(v_i, v_{i + 1}) \in E$ for $i =1, 2 \ldots, k$. A \textit{closed walk} is a walk that starts and ends at the same vertex, i.e., $v_1 = v_{k + 1}$. An edge from a node to itself is called a \textit{loop}. A \textit{subgraph} of graph $G$ is a graph created from a subset of nodes and edges of $G$. The \textit{degree} of a node is defined as the number of edges entering or exiting the node. In directed graphs, the \textit{in-degree} counts the number of incoming edges and \textit{out-degree}, the number of outgoing edges.

A graph $G$ can be represented through an \textit{adjacency matrix} $\mat{A} \in \mathbb{R}^{n \times n}$:
\begin{equation}
\label{eq:adj_matrix}
	\mat{A} = (a_{ij}); \qquad a_{ij} = \left\{
	\begin{array}{@{}ll@{}}
		w_{ij}, & \text{if edge } (i, j) \text{ exists in  graph } G \\
		0, & \text{otherwise}
	\end{array}\right.
\end{equation}
where $w_{ij}$ is the weight of the edge $(i, j)$. In this paper, we focus our attention on graphs that are undirected and \textit{unweighted} (i.e., $w_{ij} = 1$ for all edges $(i, j)$) and do not contain loops or have multiple edges between nodes. Consequently, all matrices in this paper will be symmetric, binary, and with zeros along the diagonal. Notwithstanding, it is worth mentioning that our method is general and can be applied to other classes of matrices.

\subsection{Centrality measures} \label{sec:graph_metrics}

There are many \textit{node centrality} measures, and the simplest one must be the \textit{degree centrality}~\cite{freeman_centrality_1978}. The degree of a node provides a rough estimation of its importance, yet it fails to take into consideration the connectivity of the immediate neighbours with the rest of the network. Instead, let us consider the flow of information in the network.  An important node must be part of routes where the information can flow through, and thus, be able to spread information very quickly to the rest of the network. These information routes are represented as walks over the network. This is the main idea behind walk-based centralities and was formalized by Estrada and Higham~\cite{estrada_network_2010}.

In graph theory, it is well known that the entry $(\mat{A}^k)_{ij}$  counts the number of walks of length $k \geq 1$ over graph $G$ that starts at node $i$ and end at node $j$. Then, the entry $f_{ij}$ of the matrix function $f(\mat{A})$ defined as
\begin{equation}
\label{eq:power_series}
f(\mat{A}) = \sum_{k = 0}^\infty {\zeta_k \mat{A}^k}
\end{equation}
measures how easily the information can travel from node $i$ to node $j$. The entry $\mat{A}^k$ is scaled by a coefficient $\zeta_k$, such that $\zeta_k \geq \zeta_{k + 1} \geq 0$ and $\zeta_k \rightarrow 0$ when $k$ is large, in order to penalize the contribution of longer walks and ensure the convergence of the series. The two most common choices for $f(\mat{A})$ are the matrix exponential $e^{\mat{A}}$ and resolvent $(I - \gamma \mat{A})^{-1}$~\cite{estrada_network_2010}, but other matrix functions can be used as well \cite{arrigo_edge_2016,arrigo_ml_network_2021,benzi_ranking_2013}.

From the power series (\ref{eq:power_series}), Estrada defined \textit{f-subgraph centrality}~\cite{estrada_network_2010,estrada_subgraph_2005} as the diagonal of $f(\mat{A})$, that is $fSC(i) = (f(A))_{ii}$, and measures the importance of this node based on its participation in all subgraphs in the network. The sum over all nodes of the subgraph centrality has become known as the \textit{Estrada Index}~\cite{delapena_estimating_2007,estrada_characterization_2000}, which was first introduced to quantify the degree of folding in protein chains~\cite{estrada_characterization_2000}, but later extended to characterize the global structure of general complex networks~\cite{estrada_statistical-mechanical_2007,estrada_subgraph_2006}.

Later, Benzi and Klymko~\cite{benzi_ranking_2013,benzi_total_2013} introduced the concept of \textit{f-total communicability} based on the row sum of $f(\mat{A})$, ranking the nodes according to how well they can communicate with the rest of the network. Formally, the \textit{f-total communicability} is expressed as
	\begin{equation}
	\label{eq:total_comm}
	fTC(i) = (f(\mat{A}) \vec{1})_i,
	\end{equation}
where $\vec{1}$ is a vector of length $n$ with all entries set to $1$. If we consider the matrix resolvent $f(\mat{A}) = (I - \gamma \mat{A})^{-1}$, the total communicability of a node corresponds to the well-known Katz's Centrality~\cite{hubbell_inputoutput_1965,katz_new_1953}.

In the context of network science, it is common to introduce a weighting parameter $\gamma\in(0, 1)$ and work with the parametric matrix function $f(\gamma \mat{A})$. The parameter $\gamma$ can be interpreted as the \textit{inverse temperature} and accounts for external disturbances on the network~\cite{estrada_physics_2012}. Furthermore, the value of $\gamma$ is often chosen in such a way that the terms $\zeta_k (\gamma \mat{A})^k$ in (\ref{eq:power_series}) are monotonically decreasing in order to preserve the notion that the information travels faster to nearby nodes compared to those that are farther away.

\section{Randomized algorithm for matrix functions} \label{sec:algorithm}

\begin{algorithm}[t]
\caption{A Monte Carlo method adapted from \cite{benzi_analysis_2017} for computing a matrix $\mat{F}$ as an approximation of $f(\mat{A})$. $N_s$ represents the number of random walks for approximating each row and $W_c$ is the weight cutoff.} \label{code:classical_mc_full}
\begin{algorithmic}[1]
\Function{MC}{$\mat{A}$, $N_s$, $W_c$}
\State $\mat{F} = \mat{0}$
\State $\mat{T} = \left \{ t_{ij} = \dfrac{\abs{a_{ij}}}{\sum_k{\abs{a_{ik}}}} \right \}$

\For{$i = 1, 2, \dots, n$} \Comment{for each row in $\mat{A}$}

	\For{$s = 1, 2, \dots, N_s$} \Comment{for each random walk}
        \State $\ell_0 = i; W^{(0)} = \dfrac{1}{N_s}; k = 0$

		\While{$W^{(k)} > W_c W^{(0)}$} \Comment{compute the $k$-th step}
			\State $f_{i \, \ell_k} = f_{i \, \ell_k} + \zeta_{k} W^{(k)}$
			\State $\ell_{k + 1} =$ \Call{SelectNextState}{$\mat{T}$, $\ell_k$}
			\State $W^{(k + 1)} = W^{(k)} \dfrac{a_{\ell_k \ell_{k + 1}}}{t_{\ell_k \ell_{k + 1}}}$
			\State $k = k + 1$
		\EndWhile
	\EndFor
\EndFor

\State \Return $\mat{F}$
\EndFunction
\end{algorithmic}
\end{algorithm}

Ulam and von Neumann \cite{forsythe_matrix_1950} were the first to propose a Monte Carlo method for computing the matrix inverse as a way to solve linear systems, which was later refined by \cite{benzi_analysis_2017,dimov_monte_2008,dimov_new_2015}. It consists of generating random walks over the matrix $\mat{A}$ to approximate each power in the corresponding series (\ref{eq:power_series}). Starting from a row $\ell_0$, a random walk consists of a random variable $W^{(k)}$ and a sequence of states $\ell_0, \ell_1, \ldots, \ell_k$, which are obtained by randomly jumping from one row to the next. At each step $k$, the program updates $W^{(k)}$ and add the results to entry $f(\mat{A})_{i\ell_k}$ as an approximation for the $k$-th term in the series~(\ref{eq:power_series}).

The full procedure is described in Algorithm \ref{code:classical_mc_full}. The \texttt{SelectNextState} routine randomly selects an entry in row $\ell_k$ to determine which row to jump to in the next step of the random walk. The probability of choosing an entry $j$ is equal to $\mathds{P}(\ell_{k + 1} = j \mid \ell_k = i) = t_{ij}$, where $t_{ij}$ is an entry of a transition probability matrix $\mat{T}$.

The main limitation of this method is that each random walk only updates a single entry of $f(\mat{A})$ at a time, requiring a large number of walks just to estimate a single row with reasonable accuracy. Therefore, our objective is to modify this algorithm such that it samples entire rows and columns of $\mat{A}$ when approximating each term in the series (\ref{eq:power_series}), drastically reducing the number of walks necessary to achieve the desired precision.

\subsection{Mathematical description of the method} \label{sec:math_description}

In the following, we discuss how to extend the randomized matrix product algorithm proposed in \cite{drineas_fast_2001,drineas_fast_2006-I} to compute an arbitrary matrix power.

\begin{lemma}
 	Let $R_i$ and $C_j$ denote the $i$-th row and $j$-th column of $\mat{A} \in \mathbb{R}^{n \times n}$. The matrix power $\mat{A}^p$ with $p \in \mathbb{N}$ and $p \geq 2$ can be evaluated as
\begin{equation}
\label{eq:power_mat}
	\mat{A}^p = \sum_{i_2 = 1}^n{\sum_{i_3 = 1}^n{\dots \sum_{i_p = 1}^n{C_{i_2} a_{i_2 i_3} a_{i_3 i_4} \dots a_{i_{p - 1} i_p} R_{i_p}}}}.
\end{equation}
\end{lemma}
\begin{proof}

Recall that the matrix product $\mat{AB}$ with $\mat{A}, \mat{B} \in \mathbb{R}^{n \times n}$ can be expressed as a sum of outer products \cite{drineas_fast_2006-I}:
\begin{equation*}
\mat{AB} = \sum_{k = 1}^{n} C_k B_k,
\end{equation*}
where $C_k$ is the $k$-th column of $\mat{A}$ and $B_k$ is the $k$-th row of $\mat{B}$. Therefore, a power $p$ of a square matrix $\mat{A}$ can be written as
\begin{equation}
\label{eq:power_full_recursion}
	\mat{A}^p = \sum_{k = 1}^n{C_k R_k^{(p - 1)}},
\end{equation}
where $R_k^{(p - 1)}$ is the $k$-th row of $\mat{A}^{p - 1}$. For a single row,  (\ref{eq:power_full_recursion}) is reduced to
\begin{equation}
\label{eq:power_row_recursion}
	R^{(p)}_i = \sum_{k = 1}^n{a_{ik} R_k^{(p - 1)}}.
\end{equation}

Recursively applying (\ref{eq:power_row_recursion}) for the powers $p, p - 1, \dots, 2$ and then substituting the expansion in (\ref{eq:power_full_recursion}) leads to the expression in (\ref{eq:power_mat}).

\end{proof}

\qedsymbol

\begin{corollary}
	Let $f(\mat{A}): \mathbb{R}^{n \times n} \rightarrow \mathbb{R}^{n \times n}$ be a matrix function defined by the power series
	\begin{equation}
	f(\mat{A})=\zeta_0\mat{I} + \zeta_1\mat{A}+\sum_{k=2}^{\infty} {\zeta_k \mat{A}^k}=\mat{H} + \mat{U},
	\end{equation}
    where $\mat{H} = \zeta_0\mat{I} + \zeta_1\mat{A}$. Concerning matrix $\mat{U}$, it can be written as the following sum of rank-one matrices:
	\begin{equation}
	\label{eq:full_func_sum}
		\mat{U}= \sum_{k = 2}^{\infty}{ \sum_{i_2 = 1}^n{\sum_{i_3 = 1}^n{\dots \sum_{i_k = 1}^n{\zeta_k C_{i_2} a_{i_2 i_3} a_{i_3 i_4} \dots a_{i_{k - 1} i_k} R_{i_k}}}}}.
	\end{equation}

\end{corollary}

The multiple sums appearing in the definition of matrix $\mat{U}$ in (\ref{eq:full_func_sum}) can be exploited in practice to construct a probabilistic algorithm similar to \cite{benzi_analysis_2017,dimov_monte_2008}, which was originally created for computing the inverse matrix. In fact, the formal procedure is analogous, but instead of generating random scalar variables our goal here consists of generating randomly rank-one matrices governed by the following Markov chain.

\begin{definition}
Let $\{X_k:k\ge 0\}$ be a Markov chain taking values in the state space  $S~=~\{1,2,\dots,n\}$  with initial distribution
and transition probability matrix given by
\begin{eqnarray}
(i)&& \mathds{P}(X_0 = \ell_0) = p_{\ell_0}=\frac{\|C_{\ell_0}\|_2}{\sum_{k = 1}^n{\|C_k\|_2}} \label{eq:initial_prob} \\
(ii)&& \mat{T} = (t_{ij}); \quad t_{ij} = \mathds{P}(\ell_{k + 1} = j \mid \ell_k = i) = \frac{\abs{a_{ij}}}{\sum_{k = 1}^n{\abs{a_{ik}}}}.\label{definition1}
\end{eqnarray}
\end{definition}
Here, the indices $\ell_m$ denote the corresponding state reached by the Markov chain after $m$-steps. We use the initial state $\ell_0$ of the Markov chain to choose randomly a column of matrix $\mat{A}$,  $C_{\ell_0}$. After $k$-steps of the Markov chain, the state is $\ell_k$, which is used to select the corresponding row of the matrix $\mat{A}$,  $R_{\ell_k}$. During this random evolution of the Markov chain different states are visited according to the corresponding transition probability, and along this process a suitable multiplicative random variable $W^{(k)}$ is updated conveniently. Finally,  matrix $\mat{U}$ can be computed through the expected value of a given functional, as proved formally by the following Lemma.

\begin{lemma}
	\label{prop:estimator_psi}
	Let $\mat{Z}(\omega)$ be a realization of a random matrix at a point $\omega$ of the discrete sample space, defined as
	\begin{equation}
		\label{eq:random_walk_realization}
		\mat{Z}(\omega) = \sum_{k = 0}^\infty {\zeta_{k+2} \,C_{X_0(\omega)}\, W^{(k)} R_{X_k(\omega)}}.
	\end{equation}
	Here  $W^{(k)}$ is a multiplicative random variable defined recursively as
    \begin{equation} \label{weights}
	   W^{(0)} = \frac{1}{p_{\ell_0}} \qquad
	   W^{(k)} = W^{(k - 1)} \frac{a_{\ell_{k - 1} \ell_k}}{t_{\ell_{k - 1} \ell_k}}.
    \end{equation}
	Then, it holds that
	\begin{equation}
	\label{eq:estimate_full}
	\mat{U} = \mathds{E}[\mat{Z}].
	\end{equation}
\end{lemma}

\begin{proof}
    Note that $\mat{Z}(\omega)$ is obtained from equation (\ref{eq:random_walk_realization}) as a sum of independent random matrices $\mat{Y}^{(k)}(\omega)$:
    \begin{equation}
    \label{eq:sumZ}
        \mat{Z} (\omega) = \sum_{k = 0}^\infty{\zeta_{k+2} \mat{Y}^{(k)} (\omega)},
    \end{equation}
    where
    \begin{equation*}
        \mat{Y}^{(k)}(\omega) = C_{X_0(\omega)} W^{(k)} R_{X_k(\omega)} = C_{X_0(\omega)} \frac{a_{\ell_0 \ell_1} a_{\ell_1 \ell_2} \dots a_{\ell_{k - 1} \ell_k}}{p_{\ell_0} t_{\ell_0 \ell_1} t_{\ell_1 \ell_2} \dots t_{\ell_{k - 1} \ell_k}}  R_{X_k(\omega)}.
    \end{equation*}
    Let $P^{(k)}(\omega)$ be the probability of occurring an event $\omega$ consisting in a transition from $\ell_0$ to $\ell_k$ after $k$ steps. This probability turns out to be  $p_{\ell_0} t_{\ell_0 \ell_1} t_{\ell_1 \ell_2} \dots t_{\ell_{k - 1} \ell_k}$. Therefore, the expected value of the random matrix $\mat{Y}^{(k)}(\omega)$ is given by,
    \begin{align*}
    \mathds{E}[\mat{Y}^{(k)}] & = \sum_\omega {P^{(k)}(\omega) \mat{Y}^{(k)}(\omega)} \\
    & = \sum_\omega{C_{X_0(\omega)} a_{\ell_0 \ell_1} a_{\ell_1 \ell_2} \dots a_{\ell_{k - 1} \ell_k} R_{X_k(\omega)}}.
    \end{align*}
    Note that every event $\omega$ can be described by different values of $k+1$ integer indices, running from $1$ to $n$, then
    \begin{align*}
    \mathds{E}[\mat{Y}^{(k)}] & = \sum_{i_0 = 1}^n{\sum_{i_1 = 1}^n{\dots \sum_{i_{k} = 1}^n{C_{i_0} a_{i_0 i_1} a_{i_1 i_2} \dots a_{i_{k-1} i_{k}} R_{i_{k}}}}}.
    \end{align*}
    Therefore, from (\ref{eq:full_func_sum}) we conclude that $\mathds{E}[\mat{Y}^{(k)}]=\mat{A}^{k+2}$. Finally after summing all contributions coming from any number of steps,  using (\ref{eq:sumZ}) and by linearity of the expected value operator we obtain
    \begin{equation}
    \label{eq:estimate_expected_value}
    \mathds{E}[\mat{Z}] = \sum_{k = 0}^\infty {\zeta_{k+2} \mathds{E}[\mat{Y}^{(k)}]} = \sum_{k = 0}^\infty {\zeta_{k+2} \mat{A}^{k+2}} = \mat{U}.
    \end{equation}
\end{proof}

\qedsymbol

\subsection{Practical implementation of the probabilistic method} \label{sec:rand-full}

To transform Lemma \ref{prop:estimator_psi} into a practical algorithm, we must first select a finite sample size $N_s$ and then compute the expected value $\mathds{E}[\mat{Z}]$ in (\ref{eq:estimate_full}) as the corresponding arithmetic mean. Additionally, each random walk must terminate after a finite number $m$ of steps. Mathematically, this is equivalent to considering only the first $m$ terms of the power series expansion. Since some random walks may retain important information of the matrix for longer steps than others and it is very difficult to determine a priori the number of steps required for achieving a specific precision, we adopt the following termination criteria: the computation of the random walk will end when the associated weight is less than a relative threshold $W_{c}$ \cite{benzi_analysis_2017}. In other words, a random walk terminates at  step $m$ when
\begin{equation}
    \label{eq:weight_cutoff}
    W^{(m)} \leq W_{c} W^{(0)},
\end{equation}
where $W^{(m)}$ is the weight after $m$ steps and $W^{(0)}$ is the weight at the initial step of the random walk. Formally, the infinite series in (\ref{eq:estimate_expected_value}) is truncated as
\begin{equation}
\label{eq:trunc_estimator}
f(\mat{A}) \approx \mat{H} + \mat{\hat{U}}
\end{equation}
with
\begin{equation}
\label{eq:u_hat}
\mat{\hat{U}} = \frac{1}{N_s} \sum_{s = 1}^{N_s}{\sum_{k = 0}^m{\zeta_{k + 2} C_{X_0(s)} \Omega_{\ell_0 \ell_k}^{(k)} R_{X_k(s)}}}
\end{equation}
where $\Omega_{ij}^{(k)}$ corresponds to the weight $W^{(k)}$ of the random walk that began at state $i$ and arrived at state $j$ after $k$ steps. Here $s$ indicates a realization of a random walk.

Computing $\mat{\hat{U}}$ directly from (\ref{eq:u_hat}) is ill-advised due to the large number of outer products, while also being very difficult to be parallelised efficiently. Instead, let us rearrange (\ref{eq:trunc_estimator}) to a more suitable form. The random walks with the same starting column can be grouped as
\begin{equation}
  \mat{\hat{U}} = \sum_{i = 1}^{n}{C_i \left ( \frac{1}{N_i} \sum_{s = 1}^{N_i}{\sum_{k = 0}^m{\zeta_{k + 2} \Omega_{i\, \ell_k}^{(k)} R_{X_k(s)}}} \right)},
\end{equation}
where $N_i$ is the number of random walks that began at column $i$. Assuming that $N_s >> 1$, the value of $N_i$ can be estimated \textit{a priori} as $N_i \approx p_i N_s$ with $p_i = \mathds{P}(X_0 = i)$ as defined in (\ref{eq:initial_prob}).

Let $\nu_{ij}$ denote a visit to the state $j$ at the step $k$ of a random walk that started at state $i$. For each visit $\nu_{ij}$, the weight of the walk is added to the entry $q_{ij}$ of  matrix $\mat{Q}$, defined as
\begin{equation}
\label{eq:rand_funm_Q}
 q_{ij} = \frac{1}{N_{i}}  \sum_{\nu_{ij}}{\zeta_{k + 2} \Omega_{ij}^{(k)}}.
\end{equation}
Then, we can rewrite (\ref{eq:u_hat}) as
\begin{equation}
\label{eq:rand_funm_formula}
\mat{\hat{U}} = \sum_{i = 1}^{n}{\sum_{j = 1}^{n} C_i q_{ij} R_{j}} = \mat{A}\mat{Q} \mat{A}.
\end{equation}

\begin{algorithm}[t]
\caption{A probabilistic algorithm for computing the matrix $\mat{F}$ as an approximation of $f(\mat{A})$. $N_s$ represents the total number of random walks and $W_c$ the weight cutoff.} \label{code:rand_funm_full}
\begin{algorithmic}[1]
\Function{RandFunm}{$\mat{A}$, $N_s$, $W_c$}
\State $\mat{Q} = \mat{0}$
\State $\mat{T} = \left \{ t_{ij} = \dfrac{\abs{a_{ij}}}{\sum_k{\abs{a_{ik}}}} \right \}$

\For{$i = 1, 2, \dots, n$} \Comment{for each column in $\mat{A}$}

	\State $N_i = N_s \, \mathds{P}(\ell_0 = i)$	\Comment{see equation (\ref{eq:initial_prob})}

	\For{$s = 1, 2, \dots, N_i$} \Comment{for each random walk}
        \State $\ell_0 = i; W^{(0)} = \dfrac{1}{N_i}; k = 0$

		\While{$W^{(k)} > W_c W^{(0)}$} \Comment{compute the $k$-th step}
			\State $q_{i \, \ell_k} = q_{i \, \ell_k} + \zeta_{k + 2} W^{(k)}$
			\State $\ell_{k + 1} =$ \Call{SelectNextState}{$\mat{T}$, $\ell_k$}
			\State $W^{(k + 1)} = W^{(k)} \dfrac{a_{\ell_k \ell_{k + 1}}}{t_{\ell_k \ell_{k + 1}}}$
			\State $k = k + 1$
		\EndWhile
	\EndFor
\EndFor
\State $\mat{F} = \zeta_0 \mat{I} + \zeta_1 \mat{A} + \mat{A}\mat{Q}\mat{A}$
\State \Return $\mat{F}$
\EndFunction
\end{algorithmic}
\end{algorithm}

Algorithm \ref{code:rand_funm_full} describes the procedure for approximating $f(\mat{A})$ based on equations (\ref{eq:trunc_estimator}) and (\ref{eq:rand_funm_formula}). Assuming that matrix~$\mat{A}$ is sparse with $N_{nz}$ nonzero entries, Algorithm~\ref{code:rand_funm_full} requires $O(N_s m)$ operations to construct the matrix $\mat{Q}$ and $O(n N_{nz})$ to calculate the final product $\mat{AQA}$, for a total computational cost of order of $O(N_s m + n N_{nz})$. It also uses an additional $n^2$ space in memory to store the matrix $\mat{Q}$. It is possible to reduce memory consumption if the program divides rows of $\mat{Q}$ into blocks in such a way that only one block is computed at a time. At the end of each block, the program updates the matrix $f(\mat{A})$ and reuses the memory allocation for the next block.

\subsection{Diagonal of the matrix function} \label{sec:rand-diag}

Algorithm \ref{code:rand_funm_full} can be conveniently modified to compute only the diagonal of the matrix function. Let $Q_k$ denote the $k$-th row of  matrix $\mat{Q}$ defined in (\ref{eq:rand_funm_Q}). Then, the diagonal of $f(\mat{A})$ is approximated by  vector $\vec{d} = (d_i)$ as follows

\begin{equation}
\label{eq:rand_funm_formula_diag}
d_{i} = \vec{e}_{i}^\intercal \, f(\mat{A}) \, \vec{e}_i =  \zeta_0 + \zeta_1 a_{ii} + \sum_{k = 1}^n a_{ik} \, \langle Q_k, C_i \rangle,
\end{equation}
where $\langle \cdot, \cdot \rangle$ denotes the inner product and $e_i$ a vector from the canonical basis. Essentially, the program updates the value of $d_i$ immediately after the computation of row $Q_k$. In this way, only a single row of $\mat{Q}$ needs to reside in memory at a time. Naturally, if $a_{ik} \sim 0$, the program can skip the calculation of $Q_k$ in order to save computational resources. Note that multiple entries can be computed at the same time by reusing $Q_k$ and then selecting the appropriate entry $a_{ik}$ and column $C_i$.

In terms of computational cost, the computation of all diagonal entries requires $O(N_s m + N_{nz} \bar{n}_c)$ floating-point operations, where $\bar{n}_c$ denotes the average number of nonzero entries per column and consumes an additional $n$ space in memory. The algorithm described in this section will be referred to as \texttt{RandFunmDiag}.

\subsection{Action of the matrix function over a vector} \label{sec:rand-action}

Let $\vec{v}$ be a real vector in $\mathbb{R}^{n}$, our goal is to develop another algorithm based on Lemma \ref{prop:estimator_psi} for computing $f(\mat{A}) \, \vec{v}$. First, let us multiply the truncated series~(\ref{eq:trunc_estimator}) by the vector $\vec{v}$:
\begin{equation}
f(\mat{A})\vec{v} \approx  \vec{h} + \vec{\hat{u}}
\end{equation}
where
\begin{equation}
\label{eq:vec_u_hat}
\vec{\hat{u}} = \frac{1}{N_s} \sum_{s = 1}^{N_s}{\sum_{k = 0}^m{\zeta_{k + 2} C_{X_0(s)} \Omega_{\ell_0 \ell_k}^{(k)} r_{X_k(s)}}}
\end{equation}
with $r_i = \langle R_i, \vec{v} \rangle$ and $\vec{h} = \mat{H}\vec{v}$. Rearranging the series such that the random walks with the same starting column are grouped:
\begin{equation*}
\vec{\hat{u}} = \sum_{i = 1}^{n}{C_i \left ( \frac{1}{N_i} \sum_{s = 1}^{N_i}{\sum_{k = 0}^m{\zeta_{k + 2} \Omega_{i \, \ell_k}^{(k)} r_{X_k(s)}}} \right)}.
\end{equation*}

Then, the action of the matrix function $f(\mat{A})$ over the vector $\vec{v}$ can be approximated as
\begin{equation}
\label{eq:rand_funm_action_formula}
f(\mat{A})\vec{v} \approx \vec{h} + \mat{A}\vec{q}
\end{equation}
with
\begin{equation*}
\vec{q} = (q_i);
\qquad
\mat{q}_i = \frac{1}{N_i} \sum_{s = 1}^{N_i}{\sum_{k = 0}^m{\zeta_{k + 2} \Omega_{i \, \ell_k}^{(k)} r_{X_k(s)}}}.
\end{equation*}

\begin{algorithm}[t]
\caption{A probabilistic algorithm for computing $\vec{y}$ as an approximation of  $f(\mat{A})\vec{v}$. $N_s$ represents the total number of random walks and $W_c$ the weight cutoff threshold.} \label{code:rand_funm_action}

\begin{algorithmic}[1]
\Function{RandFunmAction}{$\mat{A}$, $\vec{v}$, $N_s$, $W_c$}

\State $\vec{q} = 0$
\State $\mat{T} = \left \{ t_{ij} = \dfrac{\abs{a_{ij}}}{\sum_k{\abs{a_{ik}}}} \right \}$
\State $\vec{r} = \mat{A}\vec{v}$

\For{$i = 1, 2, \dots, n$} \Comment{for each column in $\mat{A}$}

	\State $N_i = N_s \, \mathds{P}(\ell_0 = i)$	\Comment{see equation (\ref{eq:initial_prob})}

	\For{$s = 1, 2, \dots, N_i$} \Comment{for each random walk}
		\State $\ell_0 = i; W^{(0)} = \dfrac{1}{N_i}; k = 0$

		\While{$W^{(k)} > W_c W^{(0)}$} \Comment{compute the $k$-th step}
			\State $q_i = q_i + \zeta_{k + 2} \, W^{(k)} \, r_{\ell_k}$
			\State $\ell_{k + 1} =$ \Call{SelectNextState}{$\mat{T}$, $\ell_k$}
			\State $W^{(k + 1)} = W^{(k)} \, \dfrac{a_{\ell_k \ell_{k + 1}}}{t_{\ell_k \ell_{k + 1}}}$
			\State $k = k + 1$
		\EndWhile
	\EndFor
\EndFor
\State $\vec{y} = \zeta_0 \vec{v} + \zeta_1 \vec{r} + \mat{A}\vec{q}$
\State \Return $\vec{y}$
\EndFunction

\end{algorithmic}
\end{algorithm}

Algorithm \ref{code:rand_funm_action} describes the procedure for approximating $f(\mat{A})\vec{v}$ based on equation (\ref{eq:rand_funm_action_formula}) and the definition of the vector $\vec{q}$. It has a time complexity of $O(N_s m + N_{nz})$ and requires an additional $n$ space in memory to store the vector $\vec{q}$.

\subsection{Convergence of the method and numerical errors} \label{sec:convergence}

In the following, we prove the convergence of the Monte Carlo method described by equations (\ref{eq:trunc_estimator}) and (\ref{eq:rand_funm_formula}) through the following theorem:

\begin{theorem}\label{convergence:th}
Let $m$ be any positive integer. For each $k\in \{0,\ldots,m\}$, \\ let  $(\boldsymbol{\xi}^{(k)}(s))_{s \geq 1}$ be a collection of i.i.d. vector-valued random variable in $\mathbb{R}^{n^2}$ defined as $\boldsymbol{\xi}^{(k)}(s)=\zeta_{k + 2} \Omega_{\ell_0 \ell_k}^{(k)} \vect(C_{X_0(s)} R_{X_k(s)})$, and $V(s)=\sum_{k = 0}^m \boldsymbol{\xi}^{(k)}(s)$. Let $\psi= \sum_{s = 0}^{N_s} V(s)$, $\mu=\mathds{E}[V(s)]$, $\alpha=\max_{i}\{(\sum_{j = 1}^n |a_{i j}|)^2\}<1$, and $|\zeta_{k + 1}|<|\zeta_{k}|$ when $k\to\infty$. Then
\begin{equation}
\lim_{N_s\to\infty} \frac{\psi-N_s \mu}{\sqrt{N_s}}
\end{equation}
converges in distribution to a random vector distributed according to a multivariate normal distribution $N[0,\mat{\Theta}]$. Here $\mat{\Theta}=(\theta_{ij})$ is the covariance matrix, with $\theta_{ij}=Cov(v_{i}(s),v_{j}(s))$, and
$v_i(s)$ the $i$-th component of the random vector $V(s)$.
\end{theorem}

\begin{proof}
This proof is based on the proof described in Theorem 3.4 \cite{ji_convergence_2013} conveniently modified for the current numerical method. Assuming that all random walks are independently generated, then
\begin{equation}\label{variance_total}
\Var(v_i(s))=\sum_{k = 0}^m \Var(\xi_i^{(k)}(s))\leq \sum_{k = 0}^m \mathds{E}[(\xi_i^{(k)}(s))^2]\leq \sum_{k = 0}^\infty \mathds{E}[(\xi_i^{(k)}(s))^2].
\end{equation}
Here, $\xi_i^{(k)}(s)$ denotes the $i$-th component of the random vector $\boldsymbol{\xi}^{(k)}(s)$. To compute the expected value $\mathds{E}[(\xi_i^{(k)}(s))^2]$, we have to enumerate all the different transitions that occurred between a given initial state $\ell_0$ and an arbitrary final state $\ell_k$ of the Markov chain in $k$ steps, along with the corresponding probabilities. This yields,
\begin{equation}
\mathds{E}[(\xi_i^{(k)}(s))^2]=\zeta_{k + 2}^2 \sum_{i_1 = 1}^n\cdots \sum_{i_{k} = 1}^n t_{\ell_0 i_1} t_{i_1 i_2} \cdots t_{i_{k-1} i_{k}}\frac{1}{p_{\ell_0}^2}  (\Omega_{\ell_0 i_k}^{(k)})^2 (g_i^{(i_k)})^2
\end{equation}
where $\vec{g}^{(i_k)}=\vect(C_{\ell_0} R_{i_k})$ is the vector obtained after vectorizing the matrix $C_{\ell_0} R_{i_k}$ with $g_i^{(i_k)}$ as the $i$-th component. From equation (\ref{weights}), it follows that
\begin{align}
\mathds{E}[(\xi_i^{(k)}(s))^2] &= \zeta_{k + 2}^2 \sum_{i_1 = 1}^n\sum_{i_2 = 1}^n\cdots \sum_{i_{k} = 1}^n\frac{1}{p_{\ell_0}^2} \frac{(a_{\ell_0 i_1} a_{i_1 i_2} \cdots a_{i_{k-1} i_{k}})^2}{t_{\ell_0 i_1} t_{i_1 i_2} \cdots t_{i_{k-1} i_{k}}}  (g_i^{(i_k)})^2 \nonumber\\
&=\frac{\zeta_{k + 2}^2}{p_{\ell_0}^2}\sum_{i_1 = 1}^n  \frac{a_{\ell_0 i_1}^2}{t_{\ell_0 i_{1}}} \sum_{i_2 = 1}^n \frac{a_{i_{1} i_2}^2}{t_{i_{1} i_{2}}}
\cdots
\sum_{i_{k} = 1}^n \frac{a_{i_{k-1} i_k}^2}{t_{i_{k-1} i_{k}}}(g_i^{(i_k)})^2.
\end{align}
Note that $\sum_{j = 1}^n \frac{a_{i j}^2}{t_{i j}}=(\sum_{j = 1}^n |a_{i j}|)^2$ from equation (\ref{definition1}), then it holds
\begin{equation}
\mathds{E}[(\xi_i^{(k)}(s))^2] \leq \frac{\zeta_{k + 2}^2}{p_{\ell_0}^2} \alpha^k \beta.
\end{equation}
Here $\beta=\max_{j}\{(g_i^{(j)})^2\}$.
Since $\zeta_{k+2}\leq 1,\forall k$ for any matrix function of interest, from equation (\ref{variance_total}), we have
\begin{equation}
\Var(v_i(s))\leq \sum_{k = 0}^\infty \mathds{E}[(\xi_i^{(k)}(s))^2]\leq \frac{\beta}{p_{\ell_0}^2}
\sum_{k = 0}^\infty |\zeta_{k + 2}| \alpha^k \leq \frac{\beta}{\alpha^2 p_{\ell_0}^2} \sum_{j = 0}^\infty |\zeta_{j}| \alpha^j.
\end{equation}
Since
$\lim_{k\to\infty} \alpha \frac{|\zeta_{k + 1}|}{|\zeta_{k}|}< 1$, it holds that the series  $\sum_{j = 0}^\infty |\zeta_{j}| \alpha^j$ converges and therefore the variance $\Var(v_i(s))$ is bounded.

Note, however,  that for the specific case of the exponential function, \\ $|\zeta_{k + 1}|/|\zeta_{k}|=1/(k+1)$, and therefore any value of $\alpha$ is allowed to ensure convergence of the series. However, for the inverse function $\alpha< 1$ is strictly mandatory.

Since the variance is finite, the Central Limit Theorem for vector-valued random variables (see \cite{jacod_probability_2004} e.g.) guarantees that
\begin{equation}
\frac{\psi-N_s \mu}{\sqrt{N_s}} \to^{d} N[0,\mat{\Theta}].
\end{equation}
\end{proof}

\qedsymbol

Therefore, for a finite sample size $N_s$, replacing the expected value in (\ref{eq:estimate_full}) by the arithmetic mean introduces
an error which is statistical in nature and known to be distributed according to a normal distribution. From Theorem \ref{convergence:th} the standard error of this mean $\varepsilon$ can be readily estimated as $\sigma^2/\sqrt{N_s}$, being $\sigma$ the corresponding standard deviation.

\section{Numerical examples}\label{sec:results}

\begin{table}
\centering
\caption{Numerical examples for evaluating the algorithms. Here, $k = 10^3$, $M = 10^6$ and $B = 10^9$.} \label{tab:dataset}
    \begin{tabular}[t]{c c c c p{0.4\textwidth}}\toprule
        Name & Digraph? & Nodes & Edges & Description\\\midrule
        \texttt{smallworld-<n>} & No & $2^n$ & $10 \times 2^n$ & Random graph based on the Watts-Strogatz model~\cite{watts_smallworld}. Each edge has a $10\%$ rewiring chance. \\\midrule
        \texttt{kronecker-<n>} & No & $\sim 2^n$ & $\sim 16 \times 2^n$ & Kronecker graph used by the Graph500 benchmark \cite{graph500,leskovec_kronecker_2010}.  \\\midrule
        \texttt{yeast} & No & $2114$ & $4480$ & Protein interaction network for yeast~\cite{bu_topological_2003,jeong_lethality_2001,vladimir_pajek_2006}. \\\midrule
        \texttt{power-us} & No & $4941$ & $13$k & Topological representation of the power grid of the western states in the US \cite{newman_network_2013,watts_smallworld}.\\\midrule
        \texttt{internet} & No & $23$k & $96$k & Symmetrized snapshot of the structure of the internet at the level of autonomous systems circa 2006 \cite{newman_network_2013}. \\\midrule
        \texttt{cond-mat} & No & $40$k & $351$k & Collaboration network of scientists in the field of condensed matter from 1995 to 2005 \cite{newman_network_2013,newman_structure_2003}.   \\\midrule
        \texttt{twitch} & No & $168$k & $6.8$M & Social network of Twitch users in Spring 2018~\cite{leskovec_snap,rozemberczki_twitch_2021}. \\\midrule
        \texttt{stanford} & Yes & $281$k & $2.3$M & Web graph of the Stanford University domain in 2002 \cite{leskovec_snap,leskovec_community_2008}. \\\midrule
        \texttt{orkut} & No & $3.1$M & $117$M & Social network of Orkut users in 2007 \cite{leskovec_snap,mislove_measurement_2007} \\\midrule
        \texttt{uk-2005} & Yes & $39.5$M & $936$M & Web graph of the \textit{.uk} domain in 2005 \cite{boldi_ubicrawler_2004,boldi_layered_2011,boldi_webgraph_2004}.    \\\midrule
        \texttt{twitter} & Yes & $42.6$M & $1.47$B & Social network of Twitter users in 2009 \cite{leskovec_snap,yang_patterns_2011}.  \\\bottomrule
    \end{tabular}
\end{table}

To illustrate the applicability of our method, we compute the subgraph centrality and total communicability of several complex networks using the matrix exponential $e^{\gamma \mat{A}}$ with $\gamma \in [0, 1]$. Due to the random nature of the Monte Carlo algorithms, all results reported in this section correspond to the mean value of $10$ runs of the algorithm using different random seeds.

All numerical simulations were executed on a commodity server with an AMD EPYC $9554$P $64$C $3.75$GHz and $256$GB of RAM, running Fedora~38. All randomized algorithms were implemented in \texttt{C++} using OpenMP. The code was compiled with AMD~AOCC~v4.0 with the \texttt{-O3} and \texttt{-march=znver4} flags. Our implementation uses the \texttt{PCG64DXSM} \cite{oneill_pcg_2014} random number generator.

The algorithms were tested using two synthetic graphs - \texttt{smallworld} and \texttt{kronecker} - as well as a set of networks extracted from real-world data, which are described in Table \ref{tab:dataset}. Note that, before calculating the subgraph centrality and total communicability of directed graphs, their adjacency matrix $\mat{A}$  must symmetrized as
\begin{equation}
\label{eq:digraph_sym}
\mat{B} = \begin{bmatrix}  0 & \mat{A} \\ \mat{A}^\intercal & 0 \end{bmatrix}
\end{equation}
in order to split apart the outgoing and incoming edges of the graph \cite{benzi_ranking_2013}. We also remove all duplicated edges, loops and disconnected nodes from the  \texttt{kronecker} graph generated by the Graph500 code \cite{graph500}.

\subsection{Numerical errors and accuracy}

\begin{figure}[t]
\centering
\label{fig:acc_samples}
    \begin{subfigure}{0.49 \textwidth}
        \centering
        \includegraphics[width=\linewidth]{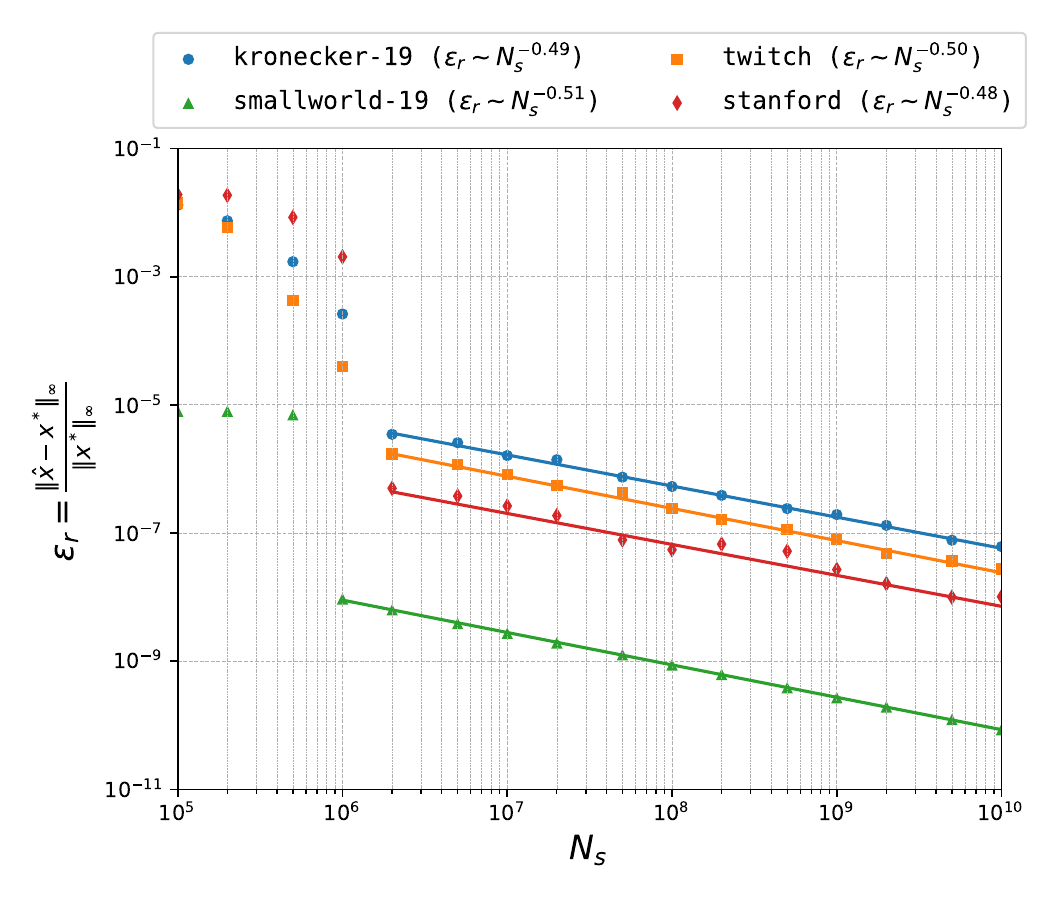}
        \caption{\texttt{RandFunmDiag} ($\gamma = 10^{-3}$).}
        \label{fig:subgraph_samples}
    \end{subfigure}
    \hfill
    \begin{subfigure}{0.49 \textwidth}
        \centering
        \includegraphics[width=\linewidth]{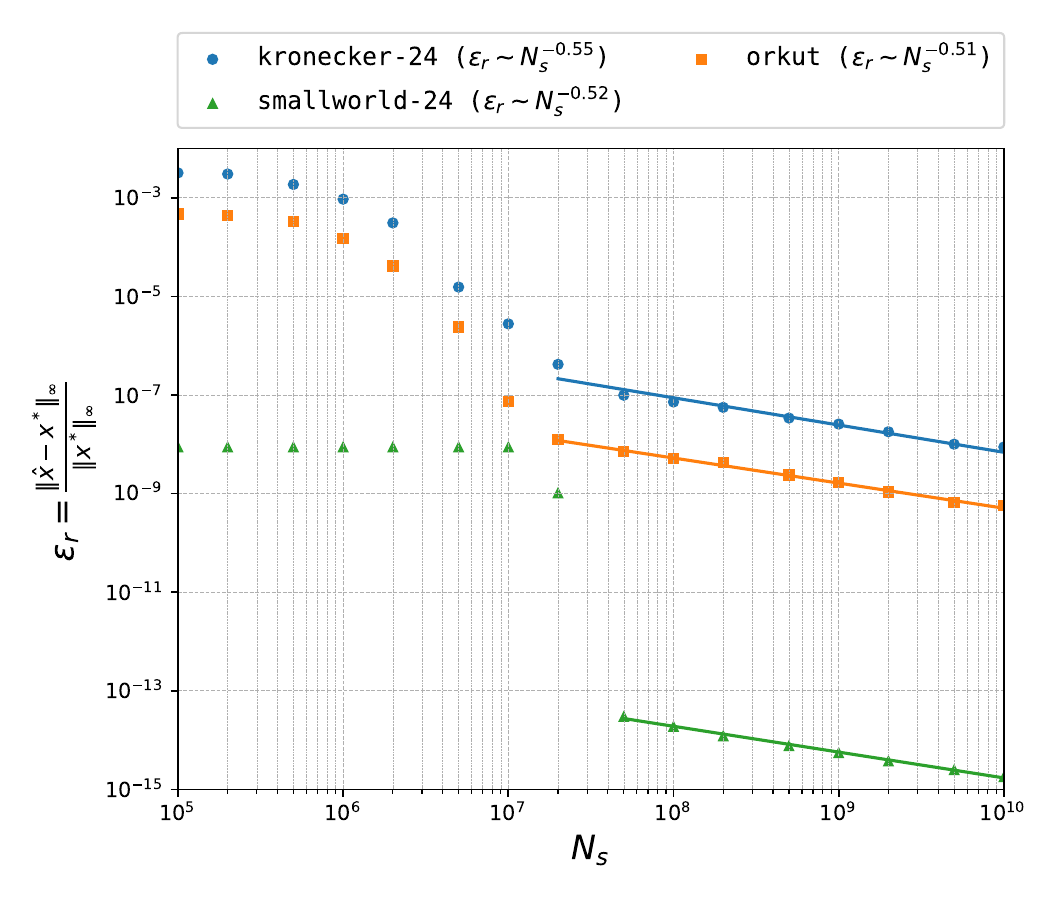}
        \caption{\texttt{RandFunmAction} ($\gamma = 10^{-5}$).}
        \label{fig:total_comm_samples}
    \end{subfigure}
    \caption{Relative $\ell_\infty$ error as function of the number of random walks $N_s$ for $W_c = 10^{-6}$. The degree $\gamma$ of the polynomial of the fitting curve is indicated as $\varepsilon_r \sim N_s^\gamma$. }
\end{figure}

\begin{table}[t]
\centering
\caption{Relative $\ell_\infty$ error for $N_s = 10^{8}$ and $W_c = 10^{-6}$ as well as the standard measurement error obtained after several independent runs.}
\label{tab:accuracy}
    \begin{subtable}[t]{0.49 \textwidth}
        \centering
        \caption{\texttt{RandFunmDiag} ($\gamma = 10^{-3}$).}
        \label{tab:subgraph_acc}
        \begin{tabular}[c]{lc} \toprule
         & $\varepsilon_r$ \\ \midrule
        \texttt{kronecker-19} & $(1.94 \pm 0.14) \times 10^{-7}$ \\\midrule
        \texttt{twitch} & $(8.09 \pm 0.67) \times 10^{-8}$ \\\midrule
        \texttt{smallworld-19} & $(2.70 \pm 0.04) \times 10^{-10}$ \\\midrule
        \texttt{stanford} & $(2.70 \pm 0.45) \times 10^{-8}$ \\\bottomrule
        \end{tabular}
    \end{subtable}
    \hfill
    \begin{subtable}[t]{0.49 \textwidth}
        \centering
        \caption{\texttt{RandFunmAction} ($\gamma = 10^{-5}$).}
        \label{tab:total_comm_acc}
        \begin{tabular}[c]{lc} \toprule
         & $\varepsilon_r$ \\ \midrule
        \texttt{kronecker-24} & $(2.57 \pm 0.26) \times 10^{-8}$\\ \midrule
        \texttt{orkut} & $(1.67 \pm 0.09) \times 10^{-9}$\\ \midrule
        \texttt{smallworld-24} & $(5.59 \pm 0.16) \times 10^{-15}$ \\\bottomrule
        \end{tabular}
    \end{subtable}
\end{table}

Figs. \ref{fig:subgraph_samples} and \ref{fig:total_comm_samples} show the relative $\ell_\infty$ error of our method as function of $N_s$. Recall from Section \ref{sec:rand-full} that the algorithm assigns \textit{a priori} $N_j \sim \|C_j \|_2$ random walks to each node $j$ of the graph. Therefore, if $N_s$ is too low, very few random walks will be assigned to a node $i$ with a low norm, such that its centrality score is basically approximated by just the first two terms in the series, i.e., $SC(i) \approx 1 +  a_{ii}$ and $TC(i) \approx 1 + (\mat{A} \vec{1})_i$. For this reason, the relative errors are highly dependent on the structure of the graph. In particular, the nodes in the  \texttt{smallworld} network have very similar probabilities, and therefore it can happen that for $N_s < n$ some nodes will have no chance to be randomly chosen. Only when the sample size is sufficiently large, the algorithm can properly estimate the centrality of every node. In this scenario, the numerical error of the algorithm scales with $O(N_s^{-0.5})$, similar to other probabilistic methods. This relation is confirmed numerically by the trend lines in Fig. \ref{fig:subgraph_samples} and \ref{fig:total_comm_samples}, which has a slope of approximately $-0.5$ in the logarithmic scale. Table~\ref{tab:accuracy} shows the relative $\ell_\infty$ error for a fixed number of samples. The table also shows the standard measurement error obtained after several independent runs.

\begin{figure}[t]
\centering
\label{fig:acc_wc}
    \begin{subfigure}{0.49 \textwidth}
        \centering
        \includegraphics[width=\linewidth]{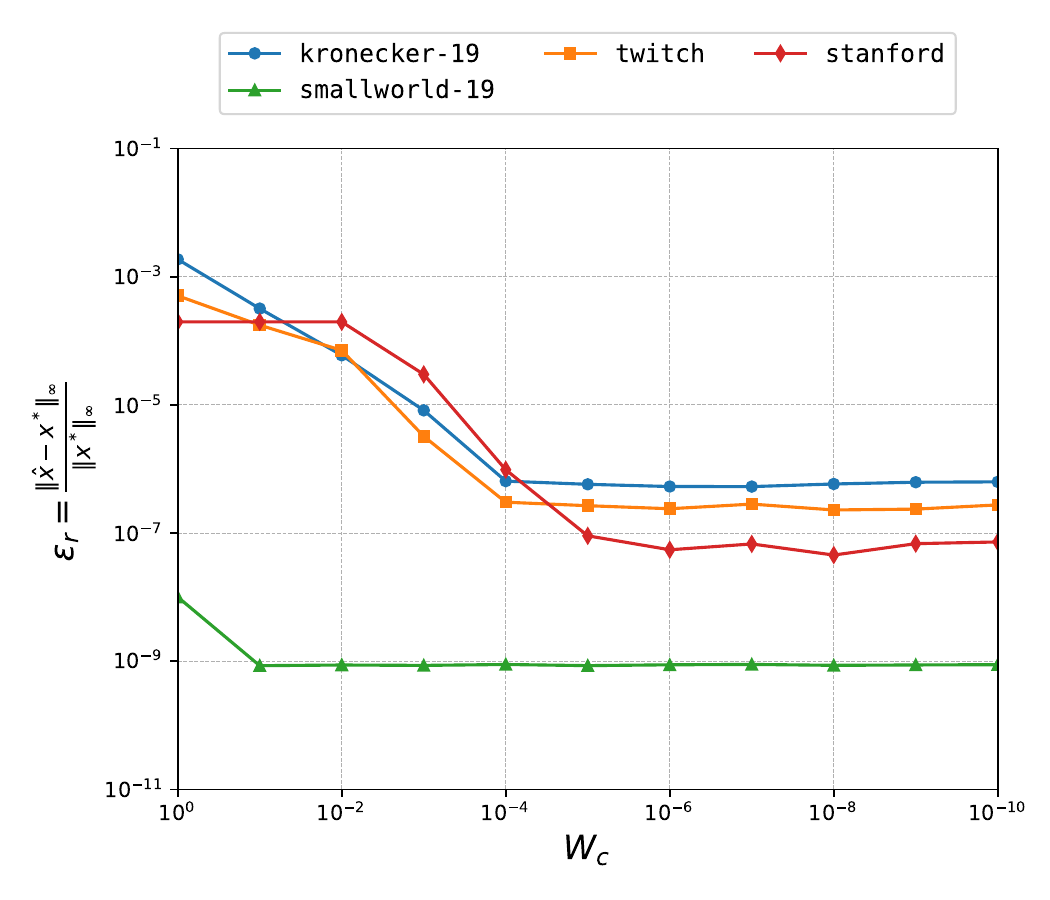}
        \caption{\texttt{RandFunmDiag} ($\gamma = 10^{-3}$).}
        \label{fig:subgraph_wc}
    \end{subfigure}
    \hfill
    \begin{subfigure}{0.49 \textwidth}
        \centering
        \includegraphics[width=\linewidth]{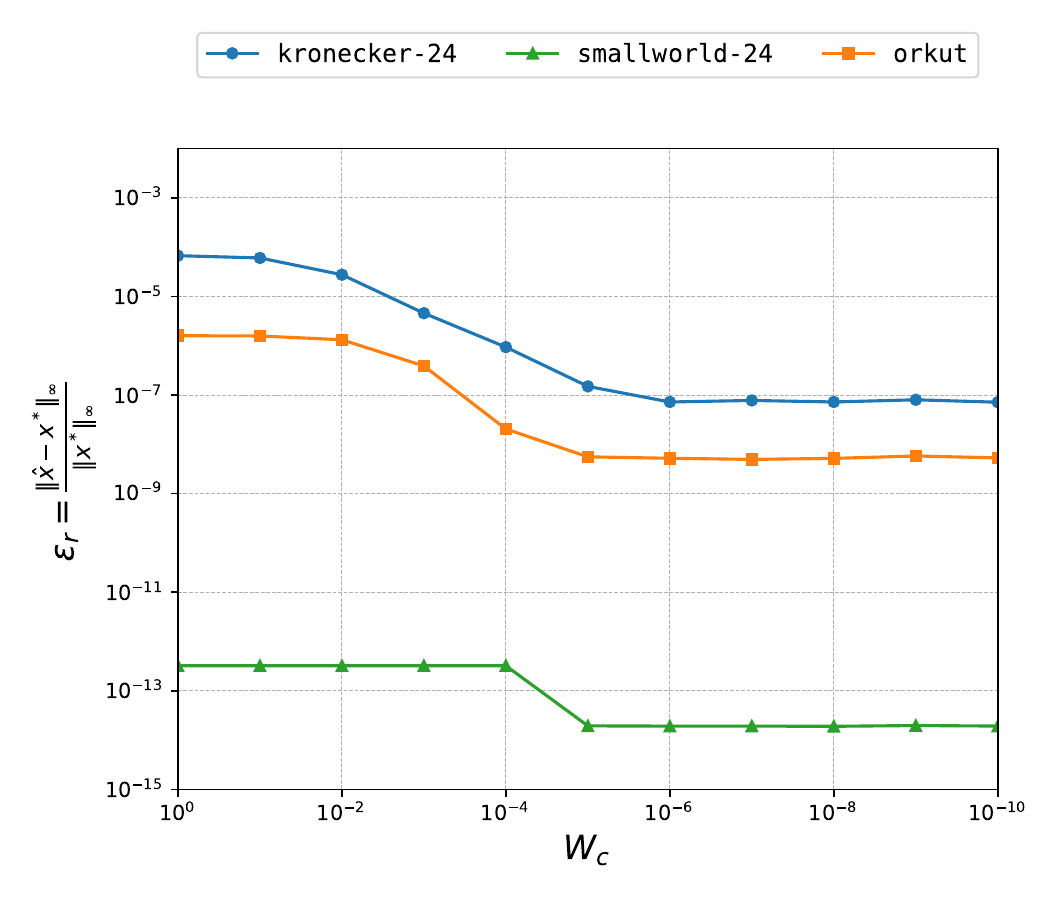}
        \caption{\texttt{RandFunmAction} ($\gamma = 10^{-5}$).}
        \label{fig:total_comm_wc}
    \end{subfigure}
    \caption{Relative $\ell_\infty$ error as function of the weight cutoff $W_c$ for $N_s = 10^{8}$.}
\end{figure}

In Figs. \ref{fig:subgraph_wc} and \ref{fig:total_comm_wc} we show the relative $\ell_\infty$ error as function of $W_c$. The value of $W_c$ controls the length of the random walks in terms of the number of steps, which is related to the number of terms of the power series expansion. When the value of $W_c$ is large, the algorithm stops the random walk generation too early, leading to large errors. On the other hand, when the value of $W_c$ is small, the algorithm continues generating the random walk for longer steps than necessary. Note that this increases the computational cost without necessarily improving the accuracy of the method. In fact, we have to consider also the statistical error, which depends on the number of generated random walks. According to Figs. \ref{fig:subgraph_wc} and \ref{fig:total_comm_wc}, the optimal value for $W_c$ is around $10^{-6}$ for these networks.

\begin{figure}[t]
\centering
\label{fig:acc_wc}
    \begin{subfigure}{0.49 \textwidth}
        \centering
        \includegraphics[width=\linewidth]{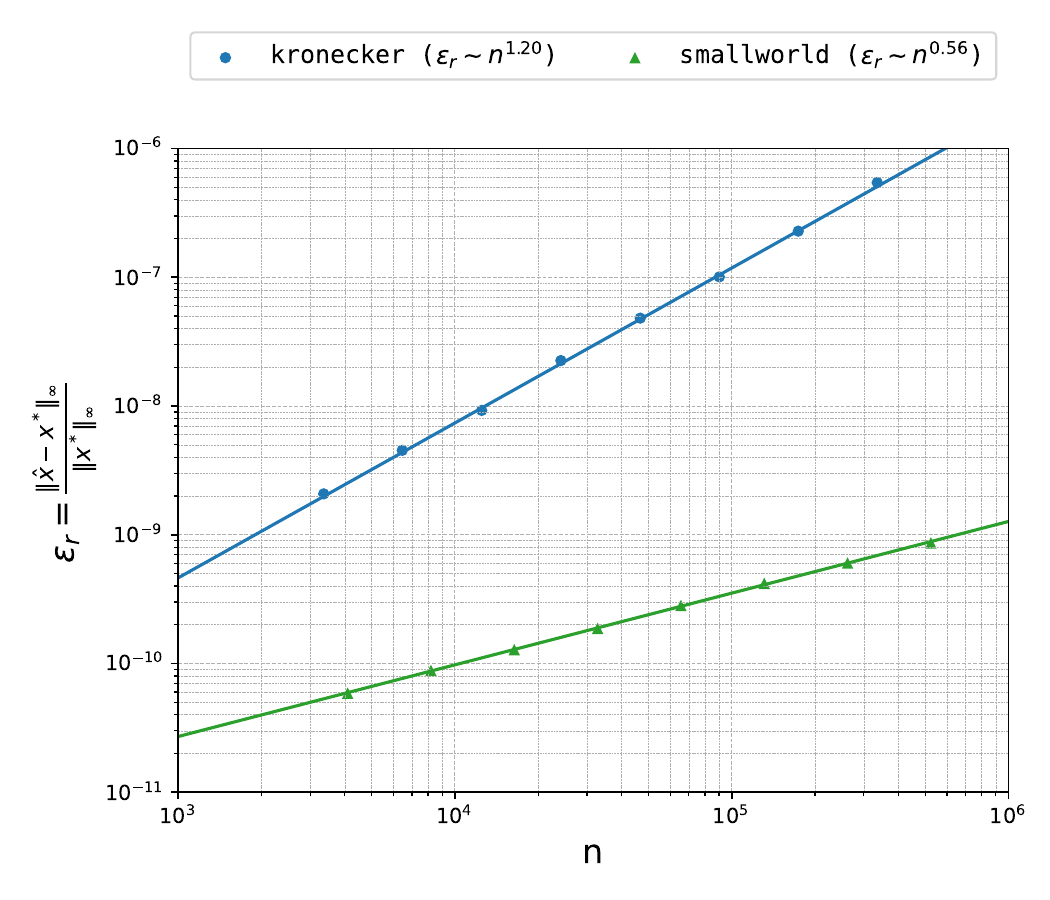}
        \caption{\texttt{RandFunmDiag} ($\gamma = 10^{-3}$).}
        \label{fig:subgraph_acc_scale}
    \end{subfigure}
    \hfill
    \begin{subfigure}{0.49 \textwidth}
        \centering
        \includegraphics[width=\linewidth]{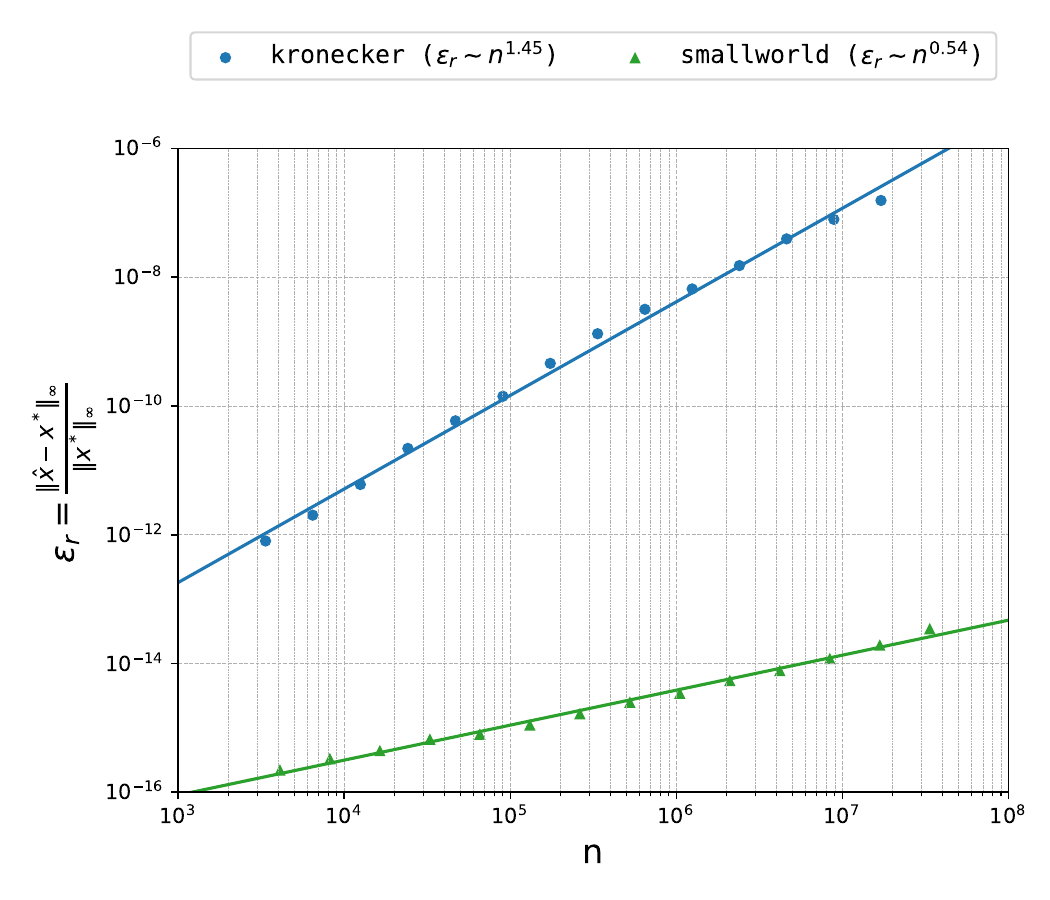}
        \caption{\texttt{RandFunmAction} ($\gamma = 10^{-5}$).}
        \label{fig:total_comm_acc_scale}
    \end{subfigure}
    \caption{Relative $\ell_\infty$ error as function of the number of nodes $n$ of the graph considering $W_c = 10^{-6}$ and $N_s = 10^{8}$. The degree $\gamma$ of the polynomial of the fitting curve is indicated as $\varepsilon_r \sim n^\gamma$.}
\end{figure}

Figs. \ref{fig:subgraph_acc_scale} and \ref{fig:total_comm_acc_scale} shows how the relative $\ell_\infty$ error grows with the size of the \texttt{smallworld} and \texttt{kronecker} networks. The \texttt{smallworld} network starts as a ring lattice with $n$ nodes, each connected to $10$ neighbours. The algorithm then rewires each edge in the lattice with a $10\%$ chance, i.e., the edge $(i, j)$ is replaced by $(i, k)$ where $k$ is chosen at random from all the possible nodes in the network such that there are no loops or duplicated edges. As a result, the random walks have very similar weights $W^{(k)}$ independently of the sequence of nodes visited. In other words, the norm of the covariance matrix $\sigma^2 = \|\mat{\Theta}\|_\infty$ is very low.  Considering that $N_s$ is fixed, there are fewer random walks to estimate the centrality score of each node as the graph size increases, which reduces the precision of the algorithm by a factor of $\sqrt{n}$, as shown in Figs. \ref{fig:subgraph_acc_scale} and \ref{fig:total_comm_acc_scale}. This is in line with the theoretical results from Section \ref{sec:convergence}, where $\varepsilon_r \sim \sigma^2 N_s^{-0.5}$.

In contrast, the nodes in the \texttt{kronecker} graph are organized hierarchically. At the top level, there is a single node acting as the central hub for the entire graph. As we move down the hierarchy, there are more nodes per level, but they have fewer connections. The number of levels in the hierarchy as well as the number of nodes and connections at each level are determined by the size of the graph. Therefore, the weight of the random walks can vary greatly depending on which nodes are visited and their position in the hierarchy. Larger graphs have a higher covariance norm $\sigma^2$ due to a wider degree difference between nodes.

\subsection{Comparison with other methods}

There are a few algorithms available in the literature for computing the matrix exponential. Perhaps the most well-known scheme is the \texttt{expm} routine from MATLAB \cite{al-mohy_new_2010,higham_scaling_2005,higham_functions_2008}. The method first scales the matrix $\mat{A}$ by a power of $2$ to reduce the norm to order $1$, calculates the Pad\'e approximant of the matrix exponential and then repeatedly squares the result to recover the original exponent. For a generic $n \times n$ matrix, \texttt{expm} requires $O(n^3)$ arithmetic operations and an additional $O(n^2)$ space in memory. \texttt{expm} calculate the entire matrix~$e^{\gamma \mat{A}}$.

To rank the nodes using the subgraph centrality, we only need to calculate the diagonal entries of $e^{\gamma \mat{A}}$, not the complete matrix. Methods for estimating the individual entries of the matrix function have been proposed by Golub, Meurant and others~\cite{benzi_quadrature_2010,fenu_block_2013,golub_matrices_2009} and are based on  Gaussian quadrature rules and the Lanczos algorithm. They require $O(n)$ operations to determine each diagonal entry, resulting in a total cost of $O(n^2)$ to calculate the subgraph centrality for all nodes. In practice, this method may suffer from numerical breakdowns when $\mat{A}$ is large and sparse~\cite{bai_able_1999,fenu_block_2013}. For this reason, it is often restricted to estimate only the $k$ most important nodes in the graph.

Likewise, the total communicability only requires the action of $f(\mat{A})$ over a vector setting $\vec{v} = \vec{1}$, which can be computed efficiently using either a polynomial or rational Krylov method \cite{afanasjew_implementation_2008,eiermann_restarted_2006,guttel_rational_2013,guttel_limitedmemory_2020}. These methods consist in generating a Krylov basis using the input matrix and then evaluating the function over the projected matrix through some direct method, such as \texttt{expm}. Assuming a sparse matrix with $N_{nz}$ nonzeros and a Krylov basis with $m$ vectors, the computational cost is $O(m N_{nz})$. In particular, we compared our method against the restarted polynomial Krylov \cite{afanasjew_implementation_2008,eiermann_restarted_2006} from the \texttt{funm\_kryl} toolbox~\cite{guttel_funm_kryl}.

While writing this article, G\"uttel and Schweitzer published a preprint \cite{guttel_randomized_2023} proposing two new randomized algorithms -- \texttt{sFOM} and \texttt{sGMRES} -- for estimating $f(\mat{A})\vec{v}$. Here, we focus on \texttt{sFOM} since \texttt{sGMRES} works best with Stieltjes functions and requires a numerical quadrature rule for all the other functions. \texttt{sFOM} first creates a random \textit{sketch} of the $\mat{A}$ and then uses an incomplete Arnoldi decomposition to generate a non-orthogonal Krylov basis from this \textit{sketch}. However, the basis may be ill-conditioned, and thus, for stabilizing the method, it is required to compute a thin QR decomposition (also called \textit{whitening} \cite{nakatsukasa_fast_2022}) of the basis before evaluating the matrix function over the projected matrix. The computational cost of \texttt{sFOM} is $O(N_{nz} m \log{m} + m^3)$.

Another preprint by Cortinovis, Kressner and Nakatsukasa \cite{cortinovis_speeding_2023} was also published recently proposing a different randomization strategy for the Krylov method. They propose an Arnoldi-like decomposition to iteratively build the non-orthogonal Krylov basis $\mat{V}_m$ using only random \textit{sketches} of the basis vectors. Again, the method may apply a \textit{whitening} \cite{nakatsukasa_fast_2022} to improve the condition number of the basis. Afterwards, the program solves a least-square problem to obtain the projected matrix. We will denote this algorithm as \texttt{rand\_kryl} and it has an overall computational cost of $O(N_{nz} m^2 + m^3)$.

The \texttt{MC} represents the Monte Carlo method adapted from \cite{benzi_analysis_2017} as described in Algorithm \ref{code:classical_mc_full}. Similar to our randomized algorithm, the length of the random walks depends on the weight cutoff $W_c$. It has a computational cost of $O(m N_s)$, where $m$ denotes the average number of steps in the random walk, and does not require additional space in memory. This method can be modified to calculate $f(\mat{A})\vec{v}$ or $f(\mat{A})_{ii}$ instead of the full matrix function. Note that for the latter, the method still computes the full $f(\mat{A})$, but discards all the off-diagonal entries.

\begin{figure}[t]
 \centering
 \includegraphics[width=\linewidth]{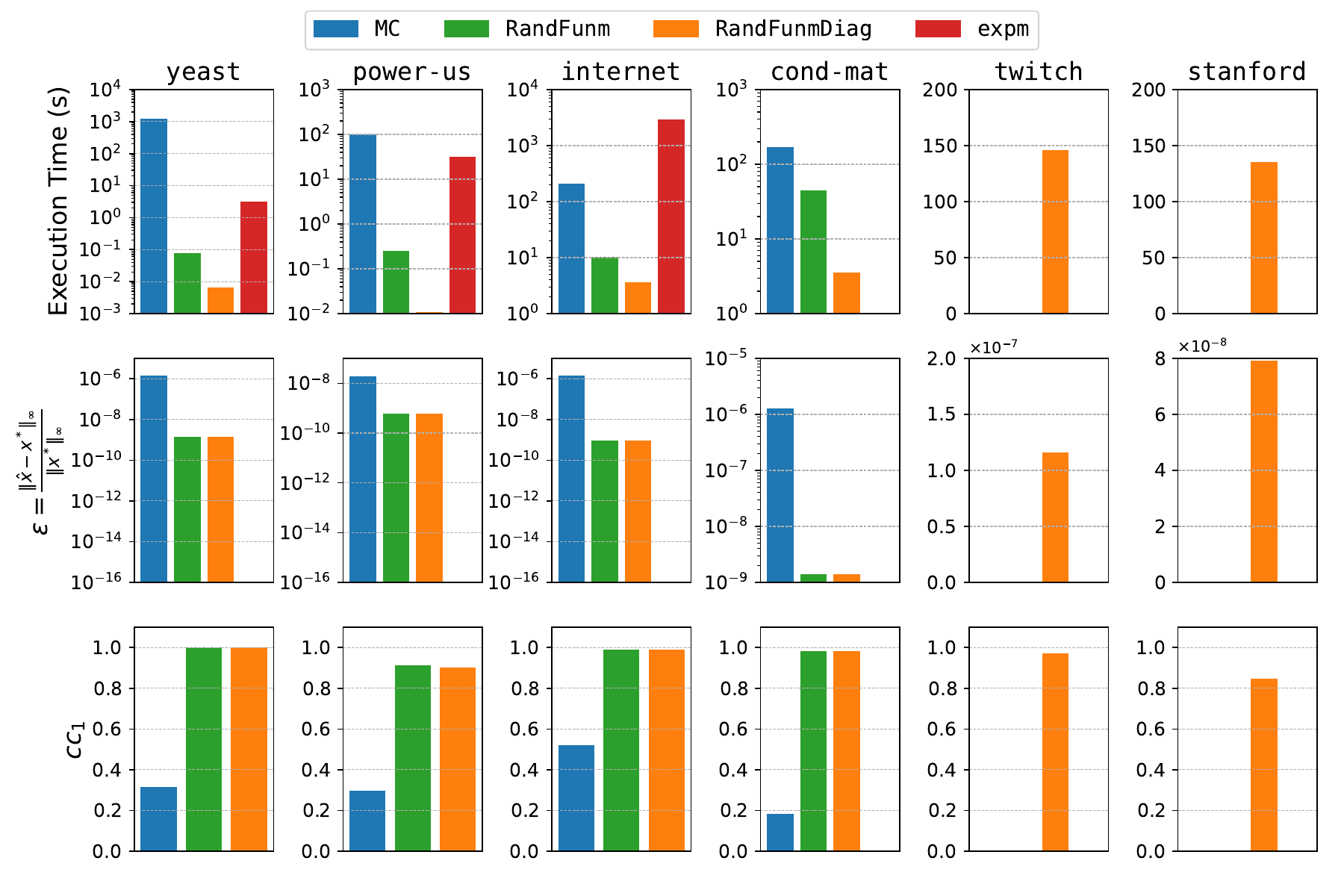}
 \caption{Comparison between different algorithms when calculating the subgraph centrality for real networks and $\gamma = 10^{-3}$. Here, we consider the centrality scores generated by our algorithm with $N_s = 10^{11}$ as reference.}
 \label{fig:subgraph_real_speed}
\end{figure}

Fig. \ref{fig:subgraph_real_speed} compares the serial execution time and accuracy among the different methods when computing the subgraph centrality. The graphs are sorted according to the number of nodes. The similarities between the two node rankings are measured using the Pearson correlation coefficient \cite{benesty_pearson_2009}. Here, $cc_1$ denotes the correlation coefficient between the top $1\%$ nodes between a reference list and the ranking obtained by the algorithm. Note that, if two or more nodes have similar centrality scores,  numerical deviations can alter their ranking order, lowering the correlation coefficient. Nevertheless, these nodes have similar importance within the network, and thus, their order in the ranking may not be relevant to understanding the dynamics of the graph.

Although \texttt{expm} reaches machine precision, it cannot be used for large graphs due to its hefty computational cost and cubic scaling. In fact, the \texttt{cond-mat} graph with $40$k nodes is already too large for \texttt{expm} and cannot be executed in a reasonable amount of time. The \texttt{MC} requires a very large number of random walks to estimate the subgraph centrality as it only updates a single entry of the matrix exponential at a time, and it is more likely for this entry to be outside the diagonal if the matrix is very large. For this reason, the accuracy of \texttt{MC} is quite poor even with a large number of random walks. Both \texttt{RandFunm} and \texttt{RandFunmDiag} have the same accuracy and correlation since the core algorithm is the same. However, \texttt{RandFunmDiag} does not require the computation of the full matrix product at the end of the algorithm, resulting in a speedup between $3$ to $24$ over \texttt{RandFunm}. Moreover, the full matrix exponential of the \texttt{twitch} and \texttt{stanford} graphs are too large to be fully represented in memory, and thus, their subgraph centrality can only be calculated by \texttt{RandFunmDiag}. The randomized algorithms also show a very high correlation for top-$1\%$ nodes in the ranking compared with the reference list.

\begin{figure}[t]
 \centering
 \includegraphics[width=\linewidth]{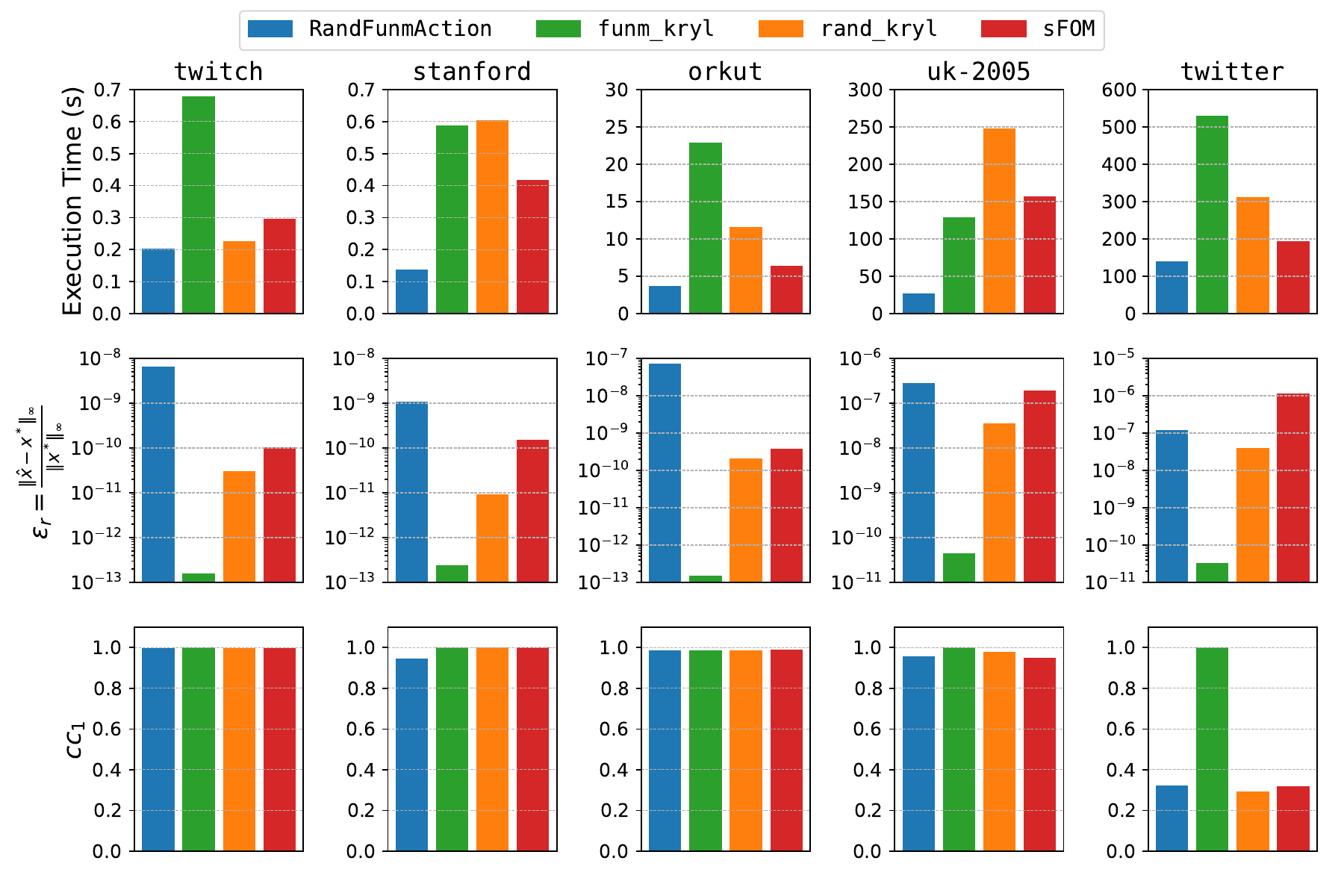}
 \caption{Comparison between different algorithms when computing the total communicability for real networks and $\gamma = 10^{-5}$. Here, we consider the centrality scores generated by \texttt{expmv}\cite{al-mohy_computing_2011} as reference.}
 \label{fig:total_comm_real_speed}
\end{figure}

Fig. \ref{fig:total_comm_real_speed} compares \texttt{RandFunmAction} against all Krylov-based methods in terms of the serial execution time and accuracy. The tolerance of \texttt{funm\_kryl} was set to $10^{-8}$ and the size of the Krylov basis of \texttt{sFOM} and \texttt{rand\_kryl} was set to $4$, such that all algorithms have comparable precision. We choose a small value of $\gamma$ to avoid overflow as we are working with large positive matrices, and consequently, all algorithms converge very fast to target precision.

The stopping criterion of \texttt{funm\_kryl} is well-known to be pessimistic \cite{guttel_limitedmemory_2020}, resulting in much higher precision than the target at the cost of higher execution times. In some networks (\texttt{twitch}, \texttt{orkut} and \texttt{twitter}), \texttt{sFOM} and \texttt{rand\_kryl} outperformed
\texttt{funm\_kryl} mainly due to the smaller basis, while in others (\texttt{uk-2005}), the additional cost associated with the \textit{sketching}, \textit{whitening}, least square QR and other operations lead to significant slowdowns. Considering the accuracy difference, the randomization in the Krylov method does not seem to be very effective when the basis is relatively small. \texttt{RandFunmAction} shows the best performance among all algorithms, in particular, for the \texttt{twitter} network, being $3.8\times$ faster than \texttt{funm\_kryl}, while \texttt{sFOM} and \texttt{rand\_kryl} are $2.7\times$ and $1.7\times$ faster, respectively.

For complex networks, it suffices for the algorithm to be sufficiently accurate to differentiate the centrality score between all nodes in the graph, there is no benefit in having higher accuracy. Indeed, the ranking produced by \texttt{RandFunmAction} has a correlation greater than $0.95$ for the top $1\%$ nodes despite having a lower accuracy than the others. The only exception is the \texttt{twitter} network. As a massive social network, there is no clear structure or hierarchy in the graph, such that many of them have similar centrality scores and small numerical variations can drastically change the rank order. In comparison, \texttt{uk-2005} is an equally large web graph that follows a more clear structure with a well-defined hub and authorities, and thus, it is less susceptible to noise. For this reason, the correlation for the \texttt{twitter} network is much lower than other graphs and requires greater precision to differentiate the nodes. Note that the top-$1\%$ in the ranking contains almost 1 million nodes in both the \texttt{twitter} and \texttt{uk-2005} networks with a wide range of centrality scores. If we consider only the top-$0.1\%$, the correlation of \texttt{RandFunmAction} for the \texttt{twitter} network increases to $0.79$.

\begin{figure}[t]
\centering
\label{fig:subgraph_scale_speed}
    \begin{subfigure}{0.49 \textwidth}
        \centering
        \includegraphics[width=\linewidth]{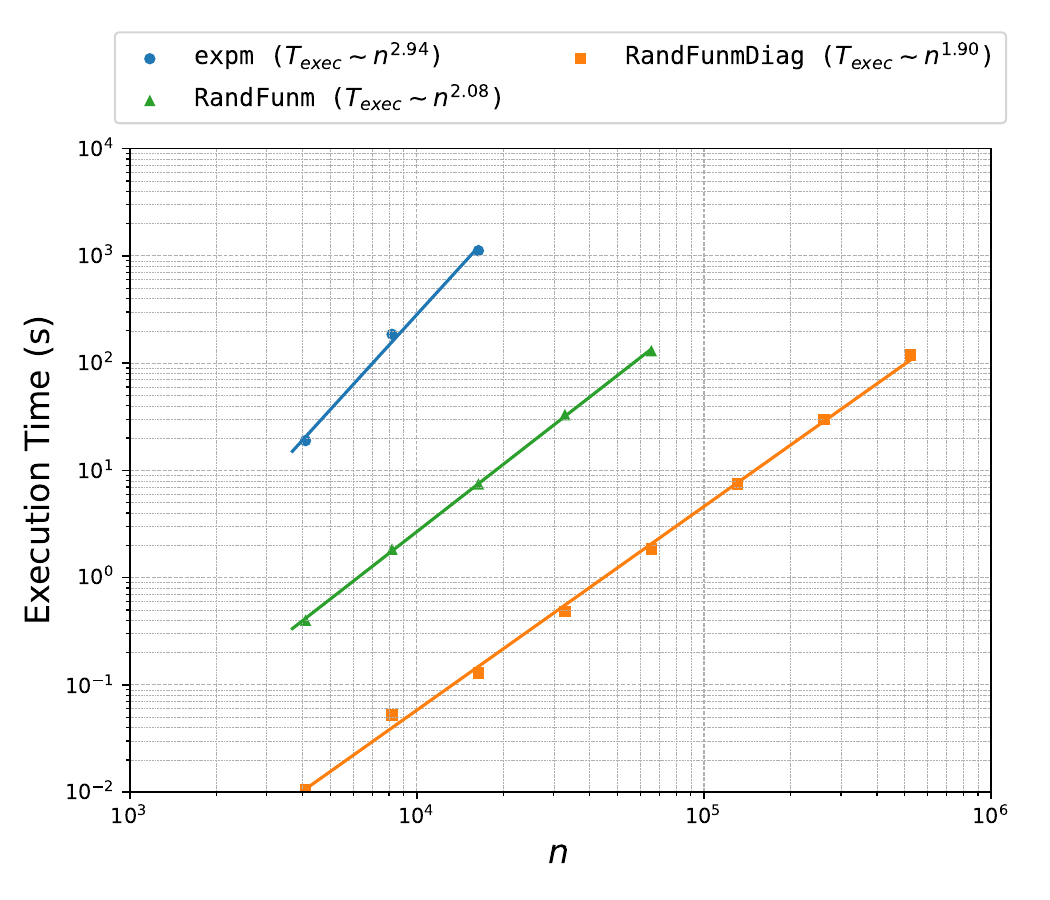}
        \caption{\texttt{smallworld} ($\varepsilon_r \leq 10^{-9}$).}
        \label{fig:subgraph_smallworld_speed}
    \end{subfigure}
    \hfill
    \begin{subfigure}{0.49 \textwidth}
        \centering
        \includegraphics[width=\linewidth]{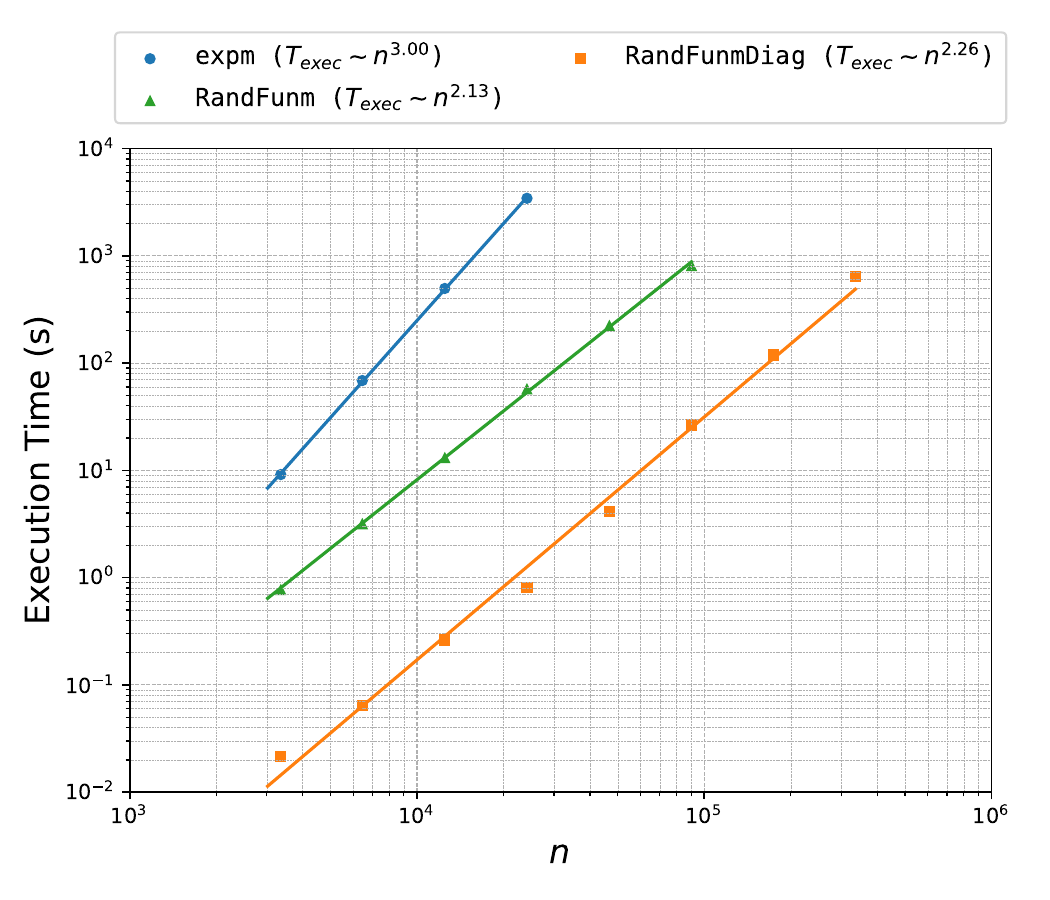}
        \caption{\texttt{kronecker} ($\varepsilon_r \leq 10^{-7}$).}
        \label{fig:subgraph_kronecker_speed}
    \end{subfigure}
    \caption{Elapsed time for computing the subgraph centrality as a function of the number of nodes $n$ for a fixed accuracy $\varepsilon_r$ and $\gamma = 10^{-3}$.}
\end{figure}

\begin{figure}[t]
\centering
\label{fig:total_comm_scale_speed}
    \begin{subfigure}{0.49 \textwidth}
        \centering
        \includegraphics[width=\linewidth]{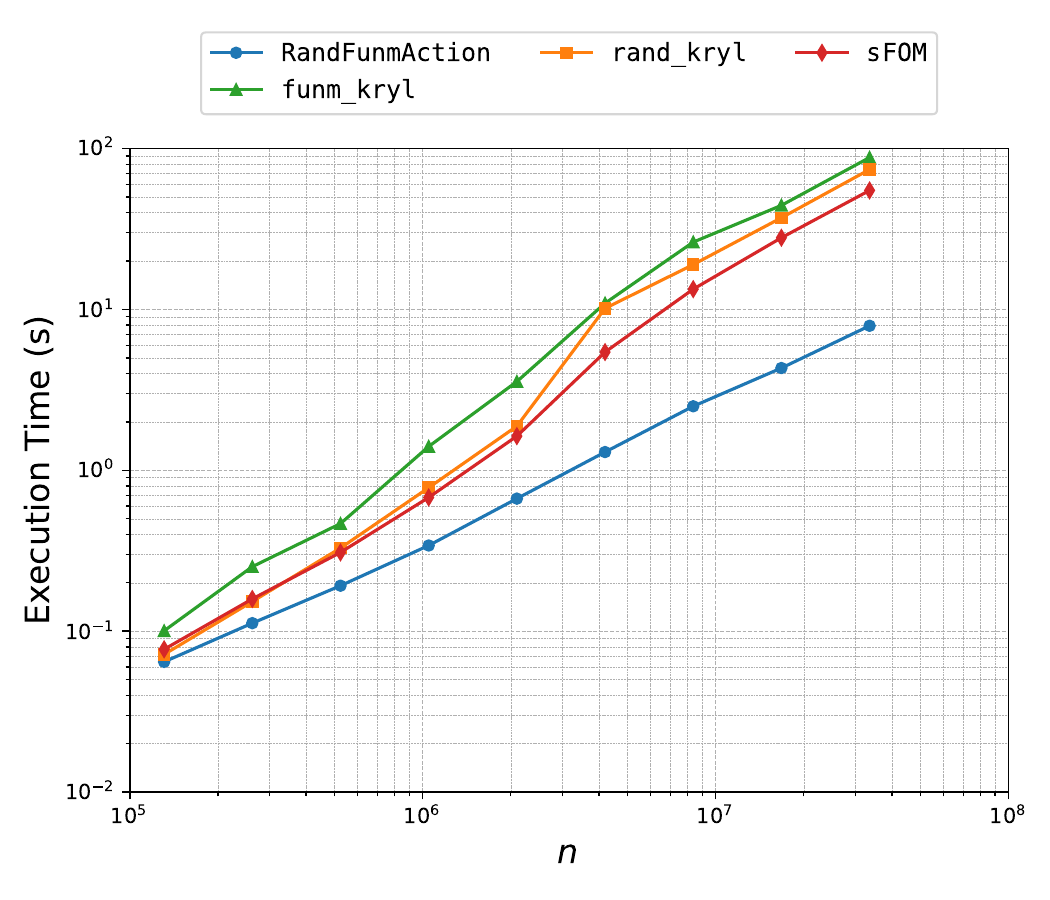}

        \caption{\texttt{smallworld} ($\varepsilon_r \leq 10^{-14}$).}
        \label{fig:total_comm_smallworld_speed}
    \end{subfigure}
    \hfill
    \begin{subfigure}{0.49 \textwidth}
        \centering
        \includegraphics[width=\linewidth]{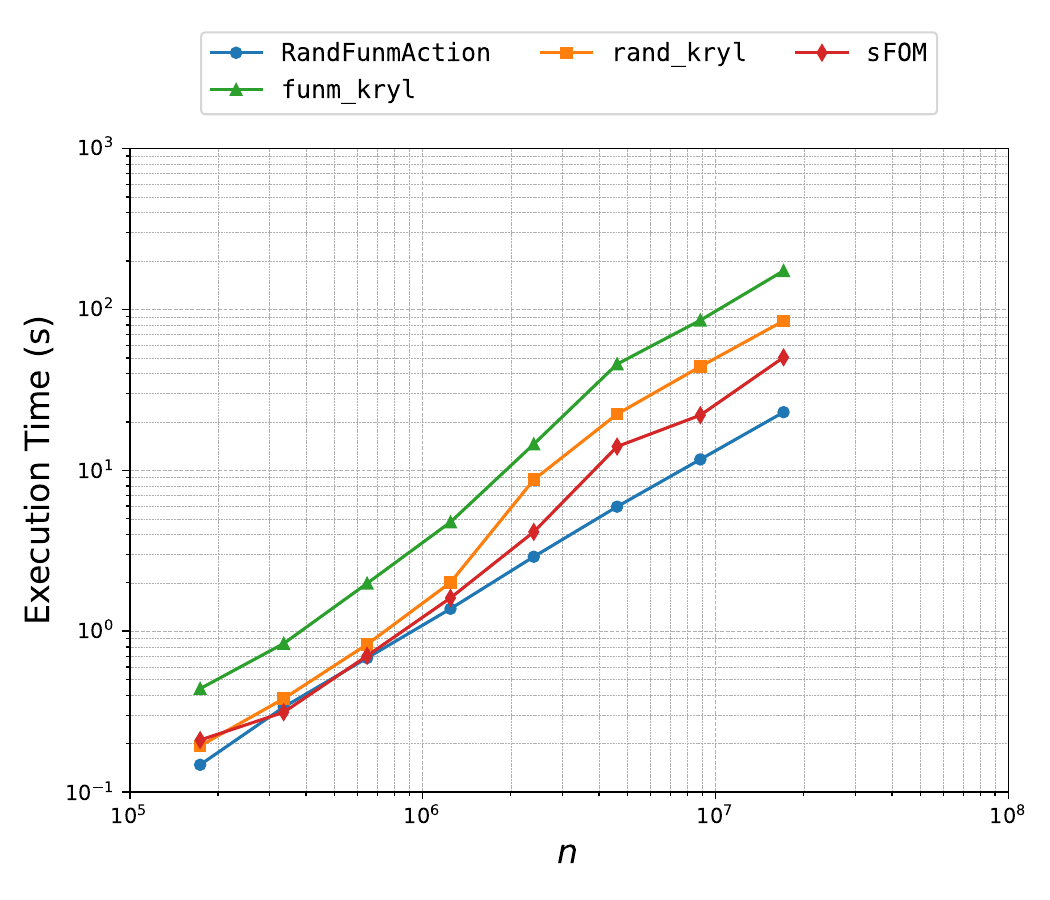}
        \caption{\texttt{kronecker} ($\varepsilon_r \leq 10^{-6}$).}
        \label{fig:total_comm_kronecker_speed}
    \end{subfigure}
    \caption{Elapsed time for computing the total communicability as a function of the number of nodes $n$ for a fixed accuracy $\varepsilon_r$ and $\gamma = 10^{-5}$.}
\end{figure}

Figs. \ref{fig:subgraph_kronecker_speed}, \ref{fig:subgraph_smallworld_speed}, \ref{fig:total_comm_kronecker_speed} and \ref{fig:total_comm_smallworld_speed} show the elapsed serial time as a function of the number of nodes for the \texttt{kronecker} and \texttt{smallworld} networks. The computational cost of the \texttt{expm} algorithm is of the order of $O(n^3)$. This is regardless of the sparsity or the distribution of the nonzeros of the matrix since it was originally proposed for dense matrices. The \texttt{MC} algorithm was too slow for the prescribed accuracy, and thus, it was not included in the graph. When $N_{nz} \sim n$ the computational cost of the \texttt{RandFunm} and \texttt{RandFunmDiag} algorithms become of order $O(n^2)$. This is because the computational cost of computing the matrix product is higher than generating random walks. Note that this cost is similar to other algorithms \cite{fenu_block_2013} for estimating diagonal entries of the matrix. However, the main advantage here is the smaller proportionality constant and the capability to compute the subgraph centrality for sparse and large matrices without worrying about a numerical breakdown \cite{fenu_block_2013}. Again, the \texttt{RandFunmDiag} algorithm is significantly faster than the \texttt{RandFunm} algorithm as it only requires the partial evaluation of the matrix product at the end of the algorithm. All Krylov-based methods scale linearly with $n$ as they rely on matrix-vector multiplications. Similar to the other Monte Carlo methods, \texttt{RandFunmAction} spends more time computing the matrix-vector product than generating the random walks, and thus, also scales linearly with $n$.

\subsection{Single entry}

\begin{table}
\centering
\caption{Relative error for calculating the total communicability of a single node $i = \argmax_j \sum_k |a_{jk}|$ considering $N_s = 10^8$, $W_c = 10^{-6}$ and $\gamma = 10^{-5}$.}
\label{tab:single_entry_acc}
\begin{tabular}[c]{lccc} \toprule
 & \texttt{stanford} & \texttt{orkut} & \texttt{kronecker-24}\\\midrule
\texttt{MC} & $1.48 \times 10^{-9}$ & $6.59 \times 10^{-8}$ & $1.33 \times 10^{-6}$ \\\midrule
\texttt{RandFunmAction} & $4.64 \times 10^{-11}$ & $3.87 \times 10^{-10}$ & $2.01 \times 10^{-8}$ \\\bottomrule
\end{tabular}
\end{table}

One of the main advantages of Monte Carlo algorithms is the ability to calculate a single entry of the solution without requiring the computation of the full solution. Table \ref{tab:single_entry_acc} shows the relative $\ell_\infty$ error for calculating the total communicability of the node with the highest degree. Both \texttt{RandFunmAction} and \texttt{MC} were modified to calculate a single entry of the solution as efficiently as possible. Due to the ability to sample entire rows and columns, \texttt{RandFunmAction} produces a much better approximation for $(f(\mat{A})v)_i$ than \texttt{MC} for the same number of random walks, independently of the network.

\subsection{Parallel Performance}

\begin{figure}[t]
\centering
\label{fig:subgraph_scale_speed}
    \begin{subfigure}{0.49 \textwidth}
        \centering
        \includegraphics[width=\linewidth]{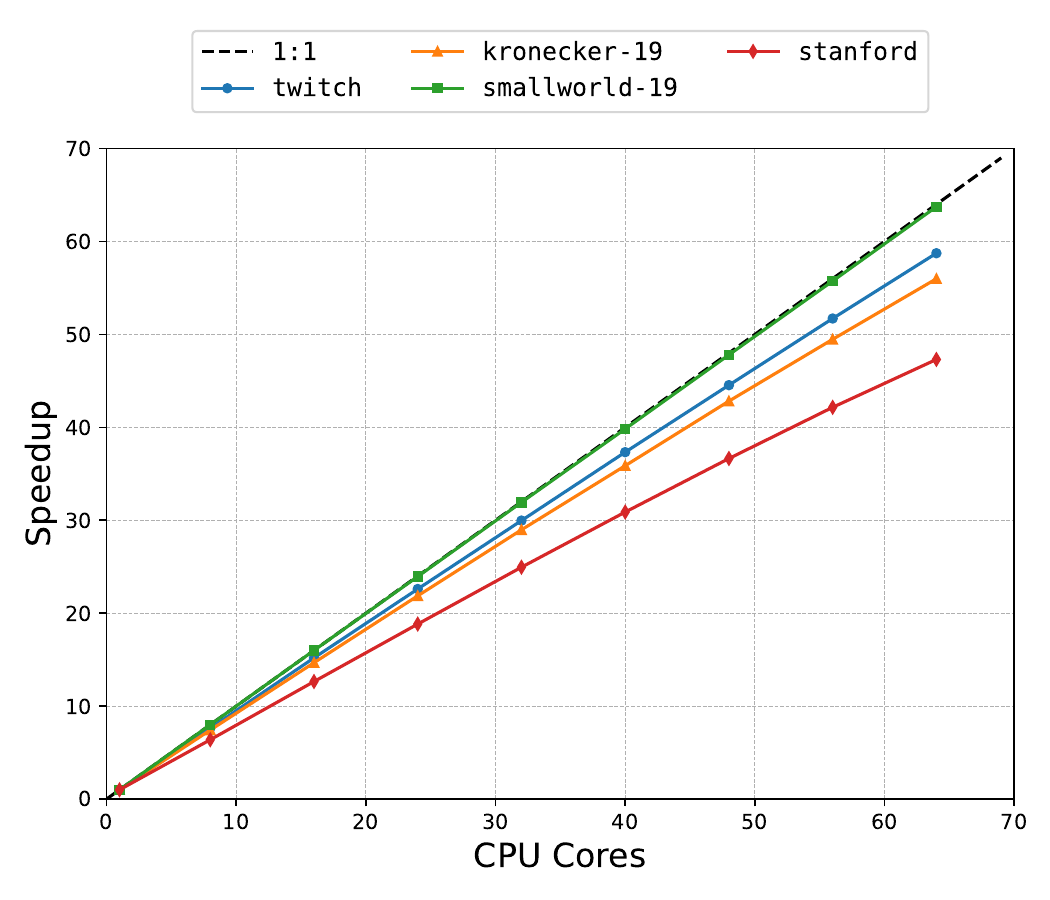}
        \caption{\texttt{RandFunmDiag} ($\gamma = 10^{-3}$).}
    \end{subfigure}
    \hfill
    \begin{subfigure}{0.49 \textwidth}
        \centering
        \includegraphics[width=\linewidth]{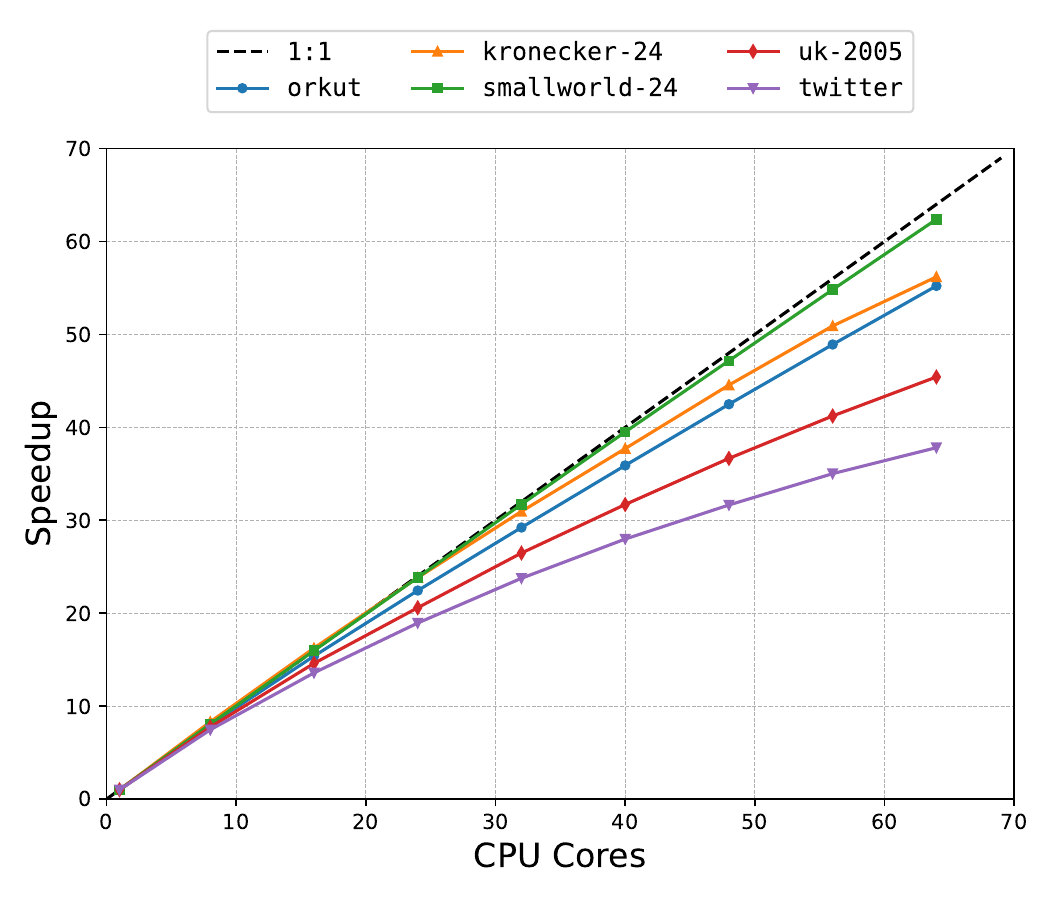}
        \caption{\texttt{RandFunmAction} ($\gamma = 10^{-5}$).}
    \end{subfigure}
    \caption{Strong scaling for $W_c = 10^{-8}$ and a fixed number of random walks.}
    \label{fig:strong_scaling}
\end{figure}

Parallelizing our randomized algorithm is fairly straightforward. Multiple rows of $\mat{Q}$ can be computed at the same time in the \texttt{RandFunmDiag} algorithm, yet the diagonal entries must be updated atomically to avoid data races. Likewise, the \texttt{RandFunmAction} algorithm can compute the vector $\vec{q}$ as well as the final product $\mat{A} \vec{q}$ completely in parallel.

Figure \ref{fig:strong_scaling} show the strong scaling for the parallel implementation of the \texttt{RandFunmDiag} and \texttt{RandFunmAction} algorithms, respectively. The scalability of both algorithms is excellent, attaining more than $85\%$ in efficiency for most networks when using $64$ cores. In particular, the parallel code was able to achieve near-perfect scaling for the \texttt{smallworld}  network due to its low degree per node and an almost uniform structure. This leads to a more efficient usage of the cache as well as an even distribution of load across the processors.

In most networks, the random walks are not distributed equally across the nodes, such that some rows of $\mat{Q}$ and entries of $\vec{q}$ take longer to compute than others. To solve this load imbalance, the program dynamically distributes the vector $\vec{q}$ and matrix $\mat{Q}$ over the CPU cores. This solution was very effective for most networks, improving significantly the performance of the program. Yet, the CPU may still be underutilized at the end of the code if the graph is very unbalanced. This is the case of directed graphs due to the symmetrization of the adjacent matrix as shown in (\ref{eq:digraph_sym}).

Another limiting factor is the latency and bandwidth of the main memory. Most operations with large and sparse matrices are well-known to be memory-bound as they cannot utilize the cache hierarchy effectively while requiring additional logic and memory accesses for handling the sparse storage format. In fact, the \texttt{kryl\_funm} algorithm shows no benefits when running in a multithreaded environment since it relies on sparse matrix-vector products and the majority of the code is written in MATLAB.  In contrast, our randomized algorithm only needs to compute two sparse matrix products: one at the beginning and another at the end of the algorithm. This still affects the scalability of the method when working with massive networks, such as \texttt{twitter} and \texttt{uk-2005}. Even under these conditions, the program was able to obtain significant speedups when using $64$ cores, achieving $60\%$ efficiency for \texttt{twitter} and $70\%$ for \texttt{uk-2005}.

\subsection{Katz Centrality} \label{sec:katz}

\begin{figure}[t]
 \centering
 \includegraphics[width=\linewidth]{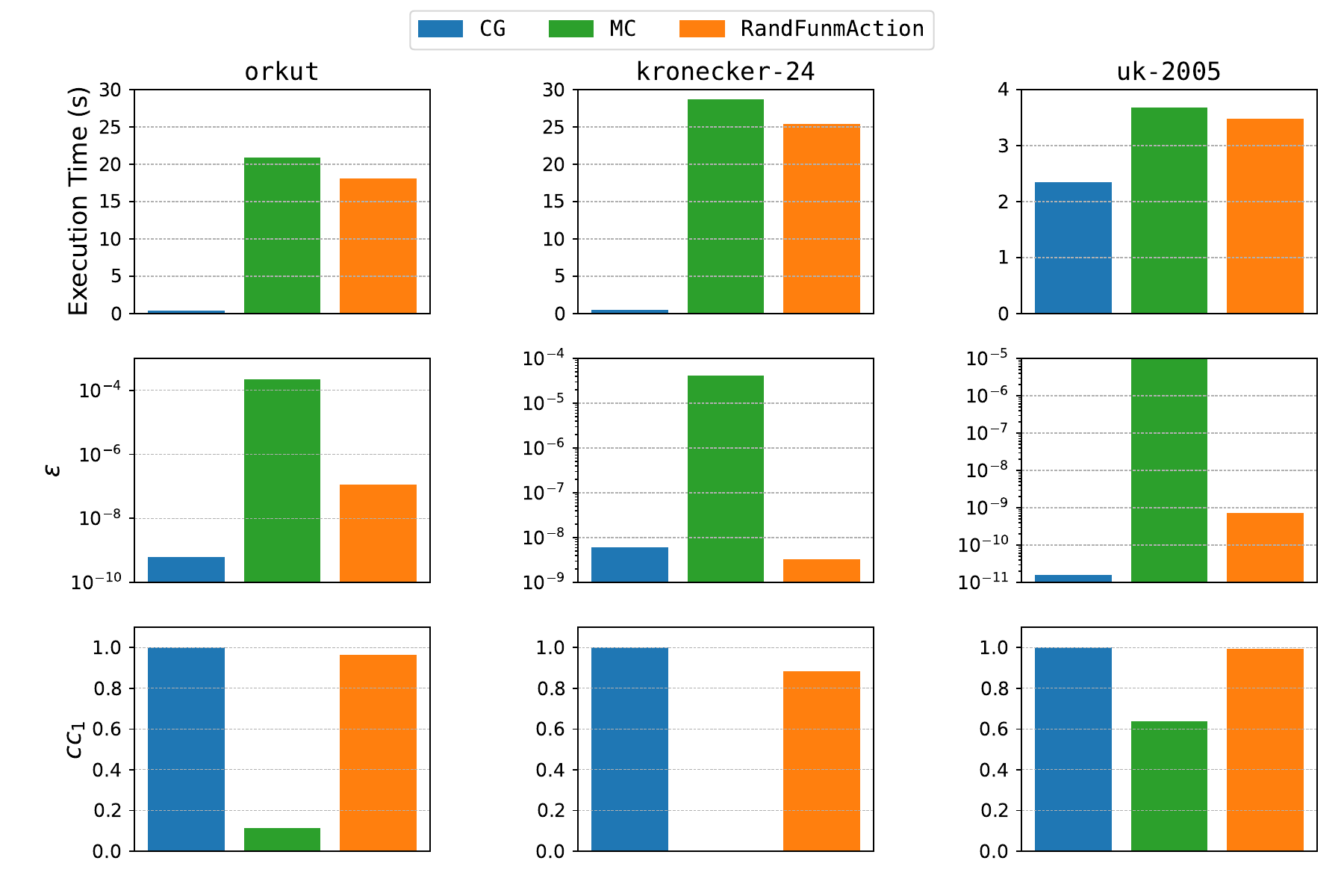}
 \caption{Comparison among the different algorithms when computing the katz-centrality for $\gamma = 0.85 \max_i{\sum_k |a_{ik}|}$ and $N_s = 10^9$. The execution time is measured using all $64$ cores. }
 \label{fig:katz-centrality}
\end{figure}

One of the most well-known centrality measures based on matrix functions is Katz's Centrality (KC) \cite{hubbell_inputoutput_1965,katz_new_1953}. It is defined as $(\mat{I} - \gamma\mat{A}) \vec{x} = \vec{1}$ with $KC(i) = x_i$ as the centrality score of the node $i$. Here, the $\gamma$ is called attenuation factor and should be between $0$ and $\rho(\mat{A})^{-1}$, where $\rho(\mat{A})$ is the spectral radius of $\mat{A}$ \cite{katz_new_1953}. Different information can be extracted from the network by changing the value of $\gamma$ \cite{benzi_limiting_2015}. For instance, if $\gamma$ tends to $\rho(\mat{A})^{-1}$, Katz's Centrality approximates the eigenvalue centrality \cite{bonacich_factoring_1972,bonacich_power_1987}. If $\gamma$ tends to $0$, then it converges to the degree centrality.

There are several ways to solve the linear system $(\mat{I} - \gamma\mat{A}) \vec{x} = \vec{1}$. Direct solvers, such as MUMPS \cite{amestoy_performance_2019,amestoy_fully_2001} or Intel Pardiso \cite{intel_mkl}, first compute the $\mat{LU}$ factorization or similar and then solve the linear system using backward/forward substitution. However, factorization is a very costly procedure, scaling with $O(n^3)$, while also requiring additional space to store matrices $\mat{L}$ and $\mat{U}$. On the other hand, sparse iterative solvers, such as Conjugate Gradient, GMRES, and BiCGSTAB, can converge very quickly to the solution especially provided a good preconditioning is available. Last, but not least, Monte Carlo methods solve the linear system as the truncated series
\begin{equation*}
(\mat{I} - \gamma \mat{A})^{-1}\vec{v} = \sum_{k = 0}^m{(\gamma \mat{A})^k\vec{v}}
\end{equation*}
with $\vec{v} = \vec{1}$. However, this series only converges when $\rho(\gamma \mat{A}) < 1$. Moreover, if $\rho(\gamma \mat{A})$ is near $1$, the convergence rate will be very slow, and thus, requires computing many terms of the expansion to reach a reasonable accuracy.

Fig. \ref{fig:katz-centrality} compares \texttt{RandFunmAction} against the original Monte Carlo method (\texttt{MC}) and a simple Conjugate Gradient algorithm (\texttt{CG}). Here, the error $\varepsilon$ of the \texttt{CG} is equal to residual norm $\|\vec{1} - \mat{A}\hat{\vec{x}}\|_2$, while for \texttt{RandFunmAction} and \texttt{MC}, it corresponds to $\ell_\infty$ error using the results from the \texttt{CG} as reference. Note that we avoid the costly computation of the eigenvalue by leveraging the Gershgorin's Theorem \cite{gershgorin_uber_1931}, i.e., $\rho(\gamma\mat{A}) \approx \max_i{\sum_k |\gamma a_{ik}|}$. Again, \texttt{RandFunmAction} is faster and more accurate than \texttt{MC} for the same number of random walks $N_s$, yet it still is not as good as the sparse iterative solver. Even without preconditioning, \texttt{CG} can converge extremely quickly to the solution due to the small value of $\gamma$. It is possible to enhance the performance of our methods by combining it with a Richardson iteration as shown in \cite{benzi_analysis_2017}. This is left for future work.

We want to emphasize that Monte Carlo methods are better suited to evaluate other matrix functions than the matrix inverse, whereas this is either too expensive or is not even possible with classical methods.

\section{Conclusion} \label{sec:conclusion}

This paper proposes a novel stochastic algorithm that randomly samples rows and columns of the matrix for approximating different powers of the power series expansion. It can evaluate any matrix function by using the corresponding coefficients of the series. The algorithm can be conveniently modified to compute either $f(\mat{A}) \vec{v}$ or the diagonal of $f(\mat{A})$ without the need to compute the entire matrix function. As a way to test the applicability of our method, we compute the subgraph centrality and total communicability of several large networks using the matrix exponential. Within this context, the stochastic algorithm has proven to be particularly effective, outperforming the competition. Our method also is highly scalable in a multithreaded environment, showing remarkable efficiency when using up to 64 cores.

In this paper, we primarily focus on the analysis of complex networks as it provided a very close relation with the method itself, but the algorithm can be applied to any scientific problem that can be expressed in terms of matrix functions, providing a quick way to estimate the solution of the problem with reasonable accuracy.

\section*{Competing interests}

All authors certify that they have no affiliations with or involvement in any organization or entity with any financial interest or non-financial interest in the subject matter or materials discussed in this manuscript.

\section*{Data availability}

The datasets generated during and/or analysed during the current study are available in the {\tt randfunm-networks} repository, available at \url{https://gitlab.com/moccalib/applications/randfunm-networks}.

\bibliographystyle{spmpsci}
\bibliography{bibliography}

\begin{thebibliography}{10}
\providecommand{\url}[1]{{#1}}
\providecommand{\urlprefix}{URL }
\expandafter\ifx\csname urlstyle\endcsname\relax
  \providecommand{\doi}[1]{DOI~\discretionary{}{}{}#1}\else
  \providecommand{\doi}{DOI~\discretionary{}{}{}\begingroup
  \urlstyle{rm}\Url}\fi

\bibitem{graph500}
Graph500.
\newblock https://graph500.org/

\bibitem{acebron_monte_2019}
Acebr{\'o}n, J.: A {{Monte Carlo}} method for computing the action of a matrix
  exponential on a vector.
\newblock Applied Mathematics and Computation \textbf{362}, 124545 (2019).
\newblock \doi{10.1016/j.amc.2019.06.059}

\bibitem{acebron_highly_2020}
Acebr{\'o}n, J.A., Herrero, J.R., Monteiro, J.: A highly parallel algorithm for
  computing the action of a matrix exponential on a vector based on a
  multilevel {{Monte Carlo}} method.
\newblock Computers \& Mathematics with Applications \textbf{79}(12),
  3495--3515 (2020).
\newblock \doi{10.1016/j.camwa.2020.02.013}

\bibitem{afanasjew_implementation_2008}
Afanasjew, M., Eiermann, M., Ernst, O.G., G{\"u}ttel, S.: Implementation of a
  restarted {{Krylov}} subspace method for the evaluation of matrix functions.
\newblock Linear Algebra and its Applications \textbf{429}(10), 2293--2314
  (2008).
\newblock \doi{10.1016/j.laa.2008.06.029}

\bibitem{al-mohy_new_2010}
{Al-Mohy}, A.H., Higham, N.J.: A {{New Scaling}} and {{Squaring Algorithm}} for
  the {{Matrix Exponential}}.
\newblock SIAM Journal on Matrix Analysis and Applications \textbf{31}(3),
  970--989 (2010).
\newblock \doi{10.1137/09074721X}

\bibitem{al-mohy_computing_2011}
{Al-Mohy}, A.H., Higham, N.J.: Computing the {{Action}} of the {{Matrix
  Exponential}}, with an {{Application}} to {{Exponential Integrators}}.
\newblock SIAM Journal on Scientific Computing \textbf{33}(2), 488--511 (2011).
\newblock \doi{10.1137/100788860}

\bibitem{albert_error_2000}
Albert, R., Jeong, H., Barab{\'a}si, A.L.: Error and attack tolerance of
  complex networks.
\newblock Nature \textbf{406}(6794), 378--382 (2000).
\newblock \doi{10.1038/35019019}

\bibitem{amestoy_performance_2019}
Amestoy, P.R., Buttari, A., L'Excellent, J.Y., Mary, T.: Performance and
  {{Scalability}} of the {{Block Low-Rank Multifrontal Factorization}} on
  {{Multicore Architectures}}.
\newblock ACM Transactions on Mathematical Software \textbf{45}(1), 1--26
  (2019).
\newblock \doi{10.1145/3242094}

\bibitem{amestoy_fully_2001}
Amestoy, P.R., Duff, I.S., L'Excellent, J.Y., Koster, J.: A {{Fully
  Asynchronous Multifrontal Solver Using Distributed Dynamic Scheduling}}.
\newblock SIAM Journal on Matrix Analysis and Applications \textbf{23}(1),
  15--41 (2001).
\newblock \doi{10.1137/S0895479899358194}

\bibitem{aparicio_lines_2022}
Aparicio, J.T., Arsenio, E., Santos, F.C., Henriques, R.: {{LINES}}:
  {{muLtImodal traNsportation rEsilience analySis}}.
\newblock Sustainability \textbf{14}(13), 7891 (2022).
\newblock \doi{10.3390/su14137891}

\bibitem{arrigo_edge_2016}
Arrigo, F., Benzi, M.: Edge {{Modification Criteria}} for {{Enhancing}} the
  {{Communicability}} of {{Digraphs}}.
\newblock SIAM Journal on Matrix Analysis and Applications \textbf{37}(1),
  443--468 (2016).
\newblock \doi{10.1137/15M1034131}

\bibitem{arrigo_ml_network_2021}
Arrigo, F., Durastante, F.: Mittag--{{Leffler Functions}} and their
  {{Applications}} in {{Network Science}}.
\newblock SIAM Journal on Matrix Analysis and Applications \textbf{42}(4),
  1581--1601 (2021).
\newblock \doi{10.1137/21M1407276}

\bibitem{bai_able_1999}
Bai, Z., Day, D., Ye, Q.: {{ABLE}}: {{An Adaptive Block Lanczos Method}} for
  {{Non-Hermitian Eigenvalue Problems}}.
\newblock SIAM Journal on Matrix Analysis and Applications \textbf{20}(4),
  1060--1082 (1999).
\newblock \doi{10.1137/S0895479897317806}

\bibitem{benesty_pearson_2009}
Benesty, J., Chen, J., Huang, Y., Cohen, I.: Pearson {{Correlation
  Coefficient}}.
\newblock In: I.~Cohen, Y.~Huang, J.~Chen, J.~Benesty (eds.) Noise
  {{Reduction}} in {{Speech Processing}}, Springer {{Topics}} in {{Signal
  Processing}}, pp. 1--4. {Springer}, {Berlin, Heidelberg} (2009).
\newblock \doi{10.1007/978-3-642-00296-0_5}

\bibitem{benzi_quadrature_2010}
Benzi, M., Boito, P.: Quadrature rule-based bounds for functions of adjacency
  matrices.
\newblock Linear Algebra and its Applications \textbf{433}(3), 637--652 (2010).
\newblock \doi{10.1016/j.laa.2010.03.035}

\bibitem{benzi_ranking_2013}
Benzi, M., Estrada, E., Klymko, C.: Ranking hubs and authorities using matrix
  functions.
\newblock Linear Algebra and its Applications \textbf{438}(5), 2447--2474
  (2013).
\newblock \doi{10.1016/j.laa.2012.10.022}

\bibitem{benzi_analysis_2017}
Benzi, M., Evans, T.M., Hamilton, S.P., Lupo~Pasini, M., Slattery, S.R.:
  Analysis of {{Monte Carlo}} accelerated iterative methods for sparse linear
  systems.
\newblock Numerical Linear Algebra with Applications \textbf{24}(3) (2017).
\newblock \doi{10.1002/nla.2088}

\bibitem{benzi_total_2013}
Benzi, M., Klymko, C.: Total communicability as a centrality measure.
\newblock Journal of Complex Networks \textbf{1}(2), 124--149 (2013).
\newblock \doi{10.1093/comnet/cnt007}

\bibitem{benzi_limiting_2015}
Benzi, M., Klymko, C.: On the {{Limiting Behavior}} of {{Parameter-Dependent
  Network Centrality Measures}}.
\newblock SIAM Journal on Matrix Analysis and Applications \textbf{36}(2),
  686--706 (2015).
\newblock \doi{10.1137/130950550}

\bibitem{boldi_ubicrawler_2004}
Boldi, P., Codenotti, B., Santini, M., Vigna, S.: {{UbiCrawler}}: {{A}}
  scalable fully distributed web crawler.
\newblock Software: Practice \& Experience \textbf{34}(8), 711--726 (2004)

\bibitem{boldi_layered_2011}
Boldi, P., Rosa, M., Santini, M., Vigna, S.: Layered label propagation: {{A
  MultiResolution}} coordinate-free ordering for compressing social networks.
\newblock In: S.~Srinivasan, K.~Ramamritham, A.~Kumar, M.P. Ravindra,
  E.~Bertino, R.~Kumar (eds.) Proceedings of the 20th International Conference
  on {{World Wide Web}}, pp. 587--596. {ACM Press}, {Hyderabad, India} (2011)

\bibitem{boldi_webgraph_2004}
Boldi, P., Vigna, S.: The {{WebGraph}} framework {{I}}: {{Compression}}
  techniques.
\newblock In: Proc. of the Thirteenth International World Wide Web Conference
  ({{WWW}} 2004), pp. 595--601. {ACM Press}, {Manhattan, USA} (2004)

\bibitem{bonacich_factoring_1972}
Bonacich, P.: Factoring and weighting approaches to status scores and clique
  identification.
\newblock The Journal of Mathematical Sociology \textbf{2}(1), 113--120 (1972).
\newblock \doi{10.1080/0022250X.1972.9989806}

\bibitem{bonacich_power_1987}
Bonacich, P.: Power and {{Centrality}}: {{A Family}} of {{Measures}}.
\newblock American Journal of Sociology \textbf{92}(5), 1170--1182 (1987).
\newblock \doi{10.1086/228631}

\bibitem{bu_topological_2003}
Bu, D., Zhao, Y., Cai, L., Xue, H., Zhu, X., Lu, H., Zhang, J., Sun, S., Ling,
  L., Zhang, N., Li, G., Chen, R.: Topological structure analysis of the
  protein{\textendash}protein interaction network in budding yeast.
\newblock Nucleic Acids Research \textbf{31}(9), 2443--2450 (2003)

\bibitem{cortinovis_speeding_2023}
Cortinovis, A., Kressner, D., Nakatsukasa, Y.: Speeding up {{Krylov}} subspace
  methods for computing f({{A}})b via randomization (2023)

\bibitem{davies_schur-parlett_2003}
Davies, P.I., Higham, N.J.: A {{Schur-Parlett Algorithm}} for {{Computing
  Matrix Functions}}.
\newblock SIAM Journal On Matrix Analysis and Applications \textbf{25}(2),
  464--485 (2003).
\newblock \doi{10.1137/S0895479802410815}

\bibitem{delapena_estimating_2007}
{de la Pe{\~n}a}, J.A., Gutman, I., Rada, J.: Estimating the {{Estrada}} index.
\newblock Linear Algebra and its Applications \textbf{427}(1), 70--76 (2007).
\newblock \doi{10.1016/j.laa.2007.06.020}

\bibitem{dimov_monte_2008}
Dimov, I.: Monte {{Carlo Methods}} for {{Applied Scientists}}.
\newblock {World Scientific}, {Singapore} (2008)

\bibitem{dimov_parallel_2001}
Dimov, I., Alexandrov, V., Karaivanova, A.: Parallel resolvent {{Monte Carlo}}
  algorithms for linear algebra problems.
\newblock Mathematics and Computers in Simulation \textbf{55}(1-3), 25--35
  (2001).
\newblock \doi{10.1016/S0378-4754(00)00243-3}

\bibitem{dimov_new_2015}
Dimov, I., Maire, S., Sellier, J.M.: A new {{Walk}} on {{Equations Monte
  Carlo}} method for solving systems of linear algebraic equations.
\newblock Applied Mathematical Modelling \textbf{39}(15), 4494--4510 (2015).
\newblock \doi{10.1016/j.apm.2014.12.018}

\bibitem{drineas_fast_2001}
Drineas, P., Kannan, R.: Fast {{Monte-Carlo}} algorithms for approximate matrix
  multiplication.
\newblock In: Proceedings 42nd {{IEEE Symposium}} on {{Foundations}} of
  {{Computer Science}}, pp. 452--459 (2001).
\newblock \doi{10.1109/SFCS.2001.959921}

\bibitem{drineas_fast_2006-I}
Drineas, P., Kannan, R., Mahoney, M.W.: Fast {{Monte Carlo Algorithms}} for
  {{Matrices I}}: {{Approximating Matrix Multiplication}}.
\newblock SIAM Journal on Computing \textbf{36}(1), 132--157 (2006).
\newblock \doi{10.1137/S0097539704442684}

\bibitem{eiermann_restarted_2006}
Eiermann, M., Ernst, O.G.: A {{Restarted Krylov Subspace Method}} for the
  {{Evaluation}} of {{Matrix Functions}}.
\newblock SIAM Journal on Numerical Analysis \textbf{44}(6), 2481--2504 (2006).
\newblock \doi{10.1137/050633846}

\bibitem{estrada_characterization_2000}
Estrada, E.: Characterization of {{3D}} molecular structure.
\newblock Chemical Physics Letters \textbf{319}(5-6), 713--718 (2000).
\newblock \doi{10.1016/S0009-2614(00)00158-5}

\bibitem{estrada_virtual_2006}
Estrada, E.: Virtual identification of essential proteins within the protein
  interaction network of yeast.
\newblock PROTEOMICS \textbf{6}(1), 35--40 (2006).
\newblock \doi{10.1002/pmic.200500209}

\bibitem{estrada_structure_2012}
Estrada, E.: The {{Structure}} of {{Complex Networks}}: {{Theory}} and
  {{Applications}}.
\newblock {Oxford University Press}, {Oxford} (2012)

\bibitem{estrada_statistical-mechanical_2007}
Estrada, E., Hatano, N.: Statistical-mechanical approach to subgraph centrality
  in complex networks.
\newblock Chemical Physics Letters \textbf{439}(1), 247--251 (2007).
\newblock \doi{10.1016/j.cplett.2007.03.098}

\bibitem{estrada_physics_2012}
Estrada, E., Hatano, N., Benzi, M.: The physics of communicability in complex
  networks.
\newblock Physics Reports \textbf{514}(3), 89--119 (2012).
\newblock \doi{10.1016/j.physrep.2012.01.006}

\bibitem{estrada_network_2010}
Estrada, E., Higham, D.J.: Network {{Properties Revealed}} through {{Matrix
  Functions}}.
\newblock SIAM Review \textbf{52}(4), 696--714 (2010).
\newblock \doi{10.1137/090761070}

\bibitem{estrada_subgraph_2005}
Estrada, E., {Rodr{\'i}guez-Vel{\'a}zquez}, J.A.: Subgraph centrality in
  complex networks.
\newblock Physical Review E \textbf{71}(5), 056103 (2005).
\newblock \doi{10.1103/PhysRevE.71.056103}

\bibitem{estrada_subgraph_2006}
Estrada, E., {Rodr{\'i}guez-Vel{\'a}zquez}, J.A.: Subgraph centrality and
  clustering in complex hyper-networks.
\newblock Physica A: Statistical Mechanics and its Applications \textbf{364},
  581--594 (2006).
\newblock \doi{10.1016/j.physa.2005.12.002}

\bibitem{fenu_block_2013}
Fenu, C., Martin, D., Reichel, L., Rodriguez, G.: Block {{Gauss}} and
  {{Anti-Gauss Quadrature}} with {{Application}} to {{Networks}}.
\newblock SIAM Journal on Matrix Analysis and Applications \textbf{34}(4),
  1655--1684 (2013).
\newblock \doi{10.1137/120886261}

\bibitem{forsythe_matrix_1950}
Forsythe, G.E., Leibler, R.A.: Matrix {{Inversion}} by a {{Monte Carlo
  Method}}.
\newblock Mathematical Tables and Other Aids to Computation \textbf{4}(31),
  127--129 (1950).
\newblock \doi{10.2307/2002508}

\bibitem{freeman_centrality_1978}
Freeman, L.C.: Centrality in social networks conceptual clarification.
\newblock Social Networks \textbf{1}(3), 215--239 (1978).
\newblock \doi{10.1016/0378-8733(78)90021-7}

\bibitem{gershgorin_uber_1931}
Gershgorin, S.: Uber die abgrenzung der eigenwerte einer matrix.
\newblock Izvestija Akademii Nauk SSSR, Serija Matematika \textbf{7}(3),
  749--754 (1931)

\bibitem{golub_matrices_2009}
Golub, G.H., Meurant, G.: Matrices, {{Moments}} and {{Quadrature}} with
  {{Applications}}.
\newblock {Princeton University Press}, {Princeton} (2009)

\bibitem{guidotti_stochastic_2023}
Guidotti, N.L., Acebr{\'o}n, J., Monteiro, J.: A {{Stochastic Method}} for
  {{Solving Time-Fractional Differential Equations}} (2023).
\newblock \doi{10.48550/arXiv.2303.15458}

\bibitem{guttel_funm_kryl}
G{\"u}ttel, S.: Funm\_kryl toolbox for {{MATLAB}}.
\newblock http://www.guettel.com/funm\_kryl/

\bibitem{guttel_rational_2013}
G{\"u}ttel, S.: Rational {{Krylov}} approximation of matrix functions:
  {{Numerical}} methods and optimal pole selection.
\newblock GAMM-Mitteilungen \textbf{36}(1), 8--31 (2013).
\newblock \doi{10.1002/gamm.201310002}

\bibitem{guttel_limitedmemory_2020}
G{\"u}ttel, S., Kressner, D., Lund, K.: Limited-memory polynomial methods for
  large-scale matrix functions.
\newblock GAMM-Mitteilungen \textbf{43}(3), e202000019 (2020).
\newblock \doi{10.1002/gamm.202000019}

\bibitem{guttel_randomized_2023}
G{\"u}ttel, S., Schweitzer, M.: Randomized sketching for {{Krylov}}
  approximations of large-scale matrix functions (2023)

\bibitem{higham_scaling_2005}
Higham, N.J.: The {{Scaling}} and {{Squaring Method}} for the {{Matrix
  Exponential Revisited}}.
\newblock SIAM Journal on Matrix Analysis and Applications \textbf{26}(4),
  1179--1193 (2005).
\newblock \doi{10.1137/04061101X}

\bibitem{higham_functions_2008}
Higham, N.J.: Functions of {{Matrices}}.
\newblock Other {{Titles}} in {{Applied Mathematics}}. {Society for Industrial
  and Applied Mathematics}, {Philadelphia, PA} (2008).
\newblock \doi{10.1137/1.9780898717778}

\bibitem{hubbell_inputoutput_1965}
Hubbell, C.H.: An {{Input-Output Approach}} to {{Clique Identification}}.
\newblock Sociometry \textbf{28}(4), 377--399 (1965).
\newblock \doi{10.2307/2785990}

\bibitem{jacod_probability_2004}
Jacod, J., Protter, P.: Probability {{Essentials}}.
\newblock Universitext. {Springer}, {Berlin, Heidelberg} (2004).
\newblock \doi{10.1007/978-3-642-55682-1}

\bibitem{jeong_lethality_2001}
Jeong, H., Mason, S.P., Barab{\'a}si, A.L., Oltvai, Z.N.: Lethality and
  centrality in protein networks.
\newblock Nature \textbf{411}(6833), 41--42 (2001).
\newblock \doi{10.1038/35075138}

\bibitem{ji_convergence_2013}
Ji, H., Mascagni, M., Li, Y.: Convergence {{Analysis}} of {{Markov Chain Monte
  Carlo Linear Solvers Using Ulam-Von Neumann Algorithm}}.
\newblock SIAM Journal on Numerical Analysis \textbf{51}(4), 2107--2122 (2013)

\bibitem{jordan_quantifying_2007}
Jord{\'a}n, F., Benedek, Z., Podani, J.: Quantifying positional importance in
  food webs: {{A}} comparison of centrality indices.
\newblock Ecological Modelling \textbf{205}(1), 270--275 (2007).
\newblock \doi{10.1016/j.ecolmodel.2007.02.032}

\bibitem{katz_new_1953}
Katz, L.: A new status index derived from sociometric analysis.
\newblock Psychometrika \textbf{18}(1), 39--43 (1953).
\newblock \doi{10.1007/BF02289026}

\bibitem{leskovec_kronecker_2010}
Leskovec, J., Chakrabarti, D., Kleinberg, J., Faloutsos, C., Ghahramani, Z.:
  Kronecker {{Graphs}}: {{An Approach}} to {{Modeling Networks}}.
\newblock The Journal of Machine Learning Research \textbf{11}, 985--1042
  (2010)

\bibitem{leskovec_snap}
Leskovec, J., Krevl, A.: {{SNAP Datasets}}: {{Stanford}} large network dataset
  collection (2014)

\bibitem{leskovec_community_2008}
Leskovec, J., Lang, K.J., Dasgupta, A., Mahoney, M.W.: Community {{Structure}}
  in {{Large Networks}}: {{Natural Cluster Sizes}} and the {{Absence}} of
  {{Large Well-Defined Clusters}} (2008).
\newblock \doi{10.48550/arXiv.0810.1355}

\bibitem{martinsson_randomized_2020}
Martinsson, P.G., Tropp, J.A.: Randomized numerical linear algebra:
  {{Foundations}} and algorithms.
\newblock Acta Numerica \textbf{29}, 403--572 (2020).
\newblock \doi{10.1017/S0962492920000021}

\bibitem{mislove_measurement_2007}
Mislove, A., Marcon, M., Gummadi, K.P., Druschel, P., Bhattacharjee, B.:
  Measurement and analysis of online social networks.
\newblock In: Proceedings of the 7th {{ACM SIGCOMM}} Conference on {{Internet}}
  Measurement, pp. 29--42. {ACM}, {San Diego California USA} (2007).
\newblock \doi{10.1145/1298306.1298311}

\bibitem{murray_randomized_2023}
Murray, R., Demmel, J., Mahoney, M.W., Erichson, N.B., Melnichenko, M., Malik,
  O.A., Grigori, L., Luszczek, P., Derezi{\'n}ski, M., Lopes, M.E., Liang, T.,
  Luo, H., Dongarra, J.: Randomized {{Numerical Linear Algebra}} : {{A
  Perspective}} on the {{Field With}} an {{Eye}} to {{Software}} (2023).
\newblock \doi{10.48550/arXiv.2302.11474}

\bibitem{nakatsukasa_fast_2022}
Nakatsukasa, Y., Tropp, J.A.: Fast \& {{Accurate Randomized Algorithms}} for
  {{Linear Systems}} and {{Eigenvalue Problems}} (2022).
\newblock \doi{10.48550/arXiv.2111.00113}

\bibitem{newman_network_2013}
Newman, M.: Network data.
\newblock http://www-personal.umich.edu/{\textasciitilde}mejn/netdata/ (2013)

\bibitem{newman_networks_2018}
Newman, M.: Networks: {{An Introduction}}.
\newblock {Oxford University Press}, {Oxford} (2018)

\bibitem{newman_structure_2003}
Newman, M.E.J.: The {{Structure}} and {{Function}} of {{Complex Networks}}.
\newblock SIAM Review \textbf{45}(2), 167--256 (2003).
\newblock \doi{10.1137/S003614450342480}

\bibitem{oneill_pcg_2014}
O'Neill, M.E.: {{PCG}}: {{A Family}} of {{Simple Fast Space-Efficient
  Statistically Good Algorithms}} for {{Random Number Generation}}.
\newblock Tech. Rep. HMC-CS-2014-0905, {Harvey Mudd College}, {Claremont, CA}
  (2014)

\bibitem{rozemberczki_twitch_2021}
Rozemberczki, B., Sarkar, R.: Twitch {{Gamers}}: A {{Dataset}} for {{Evaluating
  Proximity Preserving}} and {{Structural Role-based Node Embeddings}} (2021).
\newblock \doi{10.48550/arXiv.2101.03091}

\bibitem{vladimir_pajek_2006}
Vladimir, B., Mrvar, A.: Pajek datasets.
\newblock http://vlado.fmf.uni-lj.si/pub/networks/data/default.htm (2006)

\bibitem{intel_mkl}
Wang, E., Zhang, Q., Shen, B., Zhang, G., Lu, X., Wu, Q., Wang, Y.: Intel
  {{Math Kernel Library}}.
\newblock In: E.~Wang, Q.~Zhang, B.~Shen, G.~Zhang, X.~Lu, Q.~Wu, Y.~Wang
  (eds.) High-{{Performance Computing}} on the {{Intel}}{\textregistered}
  {{Xeon Phi}}{\texttrademark}: {{How}} to {{Fully Exploit MIC Architectures}},
  pp. 167--188. {Springer International Publishing}, {Cham} (2014)

\bibitem{watts_smallworld}
Watts, D.J., Strogatz, S.H.: Collective dynamics of `small-world' networks.
\newblock Nature \textbf{393}(6684), 440--442 (1998).
\newblock \doi{10.1038/30918}

\bibitem{yang_patterns_2011}
Yang, J., Leskovec, J.: Patterns of temporal variation in online media.
\newblock In: Proceedings of the Fourth {{ACM}} International Conference on
  {{Web}} Search and Data Mining, {{WSDM}} '11, pp. 177--186. {Association for
  Computing Machinery}, {New York, NY, USA} (2011).
\newblock \doi{10.1145/1935826.1935863}

\end{thebibliography}

\end{document}